\documentclass[a4paper,12pt]{article}
\usepackage{amssymb,amsmath,amsthm}
\usepackage{a4wide}
\usepackage{pgf,tikz}
\usetikzlibrary{arrows}
\usepackage{epsfig}
\usepackage{color}

\usepackage{subcaption}

\newtheorem{theorem}{Theorem}[section]

\newtheorem{proposition}[theorem]{Proposition}
\newtheorem{lemma}[theorem]{Lemma}
\newtheorem{corollary}[theorem]{Corollary}
\newcommand\rd{{\mathrm{d}}} 
 \newcommand\rs{{\mathrm{s}}}

\newcommand{\beq}{\begin{equation}}  
\newcommand{\eeq}{\end{equation}}

\theoremstyle{definition}

\newtheorem{remark}[theorem]{Remark}

\begin{document}


\title{Discrete Hirota 
reductions associated with the lattice KdV equation}
\author{Andrew N. W. Hone and Theodoros E. Kouloukas~\\ 
School of Mathematics, Statistics \&  Actuarial Science~\\ 
University of Kent~\\
Canterbury CT2 7FS, U.K. 
}

\maketitle

\begin{abstract}
We study the integrability of a family of birational maps obtained as reductions of the 
discrete Hirota 
equation, which are related  to travelling wave 
solutions of the lattice KdV equation. In particular, for reductions 
corresponding to waves moving with rational speed $N/M$ on the lattice, 
where $N,M$ are  coprime integers, we 
prove the Liouville integrability of the maps when $N+M$ is odd, and prove various properties of the general 
case. There are two main ingredients to our construction:  the cluster algebra associated with each of the 
Hirota bilinear equations, which provides invariant  (pre)symplectic  and Poisson structures; and the connection 
of the monodromy matrices of the dressing chain with those of the KdV travelling wave reductions.       
\end{abstract}

\section{Introduction}
The discrete Hirota 
equation \cite{zabrodin} 
(which is also known as the bilinear discrete Kadomtsev--Petviashvili equation, 
the Hirota-Miwa equation  \cite{Hirota,Miwa}, or the octahedron recurrence \cite{fz})   
is an integrable bilinear partial difference equation for 
a function $T=T(n_1,n_2,n_3)$ of three independent variables, namely 
$$ T_{n_1+1} T_{n_1-1} =T_{n_2+1} T_{n_2-1} + T_{n_3+1} T_{n_3-1}, 
$$ 
where for brevity we take $T_{n_1\pm 1}=T(n_1\pm 1,n_2,n_3)$, 
and similarly for shifts in the $n_2$ and $n_3$ directions. 
The integrable characterization of this equation is justified by its multidimensional consistency property and the existence of a Lax representation. 
Particular reductions of the discrete Hirota equation give rise to well-known integrable partial difference equations in two independent variables, as well as integrable ordinary difference equations.  
Plane wave reductions, given by  
$$ 
T (n_1,n_2,n_3) = \exp (C_1 n_1^2+C_2 n_2^2+C_3 n_3^2) \, \tau_m, \qquad 
m=n_0+\delta_1n_1+\delta_2n_2+\delta_3n_3   
$$ 
for  $\delta_i$ 
integers or half-integers, 
result in bilinear ordinary difference equations  
 of the form
\begin{equation} \label{d1d2d3}
\tau_{m+\delta_1}\tau_{m-\delta_1}=a\,  \tau_{m+\delta_2}\tau_{m-\delta_2}+ 
b\, \tau_{m+\delta_3}\tau_{m-\delta_3},
\end{equation}
with suitable constants $a,b$, 
which are recurrence relations of Gale--Robinson/Somos
type \cite{jmz, somos}. 
These kinds of recurrences inherit a Lax representation from the Lax representation of the discrete Hirota equation \cite{HKW}. 
Their non-autonomous versions are associated with $q$-Painlev\'{e} equations and their higher order analogues \cite{honeinoue, okubo, okubos}, 
and they appear in the context of supersymmetric gauge theories and 
dimer models \cite{bershtein, bershtein2, eager, gk}. 
Furthermore, they are particular examples of cluster maps, which arise from cluster mutations
of periodic quivers \cite{FM} and, as a consequence they exhibit the Laurent phenomenon, i.e.\  all iterates are Laurent polynomials in the initial data with integer coefficients \cite{fz, mase}. 
Cluster maps 
admit an invariant presymplectic form, and can be reduced to 
lower-dimensional symplectic maps \cite{FH}, which (following \cite{honeinoue})   we refer to as {\it{$U$-systems}}. 

\begin{figure} \centering 
\includegraphics[width=10cm,height=10cm,keepaspectratio]{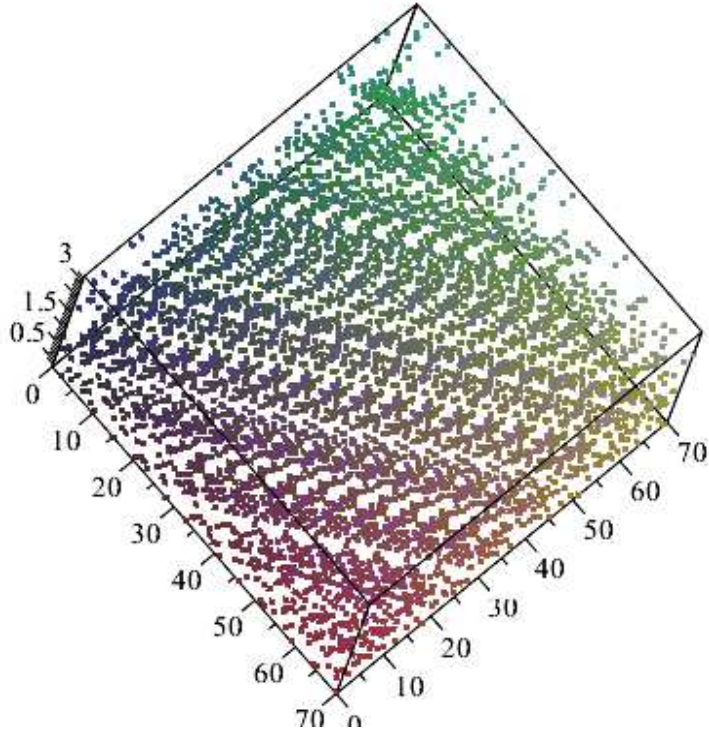}
\caption{\small{3D plot of $V_{k,l}$ against $k,l$ for a  travelling wave solution  of the discrete KdV equation (\ref{KdV}) 
given by $V_{k,l}=v_{4k-3l+1000}$ where $v_m$ satisfies (\ref{RedkdV}) with $N=4$, $M=3$, $\alpha=-1$ and 
initial data $v_j=1$ for $0\leq j\leq 5$, $v_6=3$.
  }}
\label{f1}
\end{figure}

In recent work \cite{HKQ,HKW}, two families of discrete Hirota reductions \eqref{d1d2d3} associated with 
two-dimensional lattice equations of discrete KdV/discrete Toda type 
have been studied. Both of them admit Lax representations 
which generate 
first integrals. In this paper, we focus on the 
discrete KdV family, that is  plane wave reductions of  the 
discrete Hirota equation that take the form   
\begin{equation} \label{genform}
\tau_{m+2M+N}\tau_m =a \,\tau_{m+2 M}\tau_{m+N} +  b \,\tau_{m+M+N}  \tau_{m+M} , \ \ M,N \in \mathbb{N},   
\end{equation}
where $a,b$ are constant parameters. 
These are examples of generalized T-systems \cite{kuniba}, and they are  
related  to travelling wave (periodic) reductions of 
Hirota's lattice KdV equation, namely 
\begin{equation} \label{KdV}
V_{k+1,l}-V_{k,l+1}=\alpha\left(\frac{1}{V_{k,l}}-\frac{1}{V_{k+1,l+1}}\right).
\end{equation}
Our aim is to provide the details of the Liouville integrability of the  
travelling wave reductions of (\ref{KdV}), and then use this to 
infer the integrability of the 
associated 
U-systems underlying these recurrences with respect to their corresponding symplectic structures.   
Thus we extend the 
results of \cite{HKQT}, where the particular family of $(N,1)$ travelling waves 
was considered, and give complete proofs of various assertions made 
concerning the case of general  reductions of type $(N,M)$
in \cite{HKW}.

\section{T-systems and U-systems}

There are two different classes of symplectic maps (U-systems) corresponding to \eqref{genform}, depending 
on which of the two integers $M,N$ is greater. Therefore, 
henceforth we will always make the assumption that 
$$N > M$$
and  separate  \eqref{genform} 
into the two different cases 
\begin{eqnarray}
\tau_{m+2N+M}\tau_m &=& a \,\tau_{m+2 N}\tau_{m+M} +b\, \tau_{m+N+M}  \tau_{m+N}  , \  \label{t1} \\
\tau_{m+2M+N}\tau_{m}&=& a \,\tau_{m+2 M}\tau_{m+N}+ b \,\tau_{m+N+M}  \tau_{m+M}. \  \label{t2}
\end{eqnarray}  
Furthermore, we shall also assume the coprimality condition 
$$\gcd (M,N)=1,$$ 
since otherwise the T-systems (\ref{t1}) and (\ref{t2}) 
can be decoupled  into copies of systems in lower dimension that do satisfy this condition.

The corresponding U-systems are described by the following proposition. 

\begin{proposition} \label{redsymp}
For $N+M$  odd,  
\item $\mathbf{(1)}$ $\tau_m$ satisfies (\ref{t1}) iff $u_m=\frac{\tau_{m} \tau_{m+N+1}}{\tau_{m+1} \tau_{m+N}}$ satisfies   the U-system 
\begin{equation} \label{1o}
u_{m}u_{m+1} \ldots u_{m+N+M-1}=b + a u_{m+M} u_{m+M+1} \ldots u_{m+N-1},
\end{equation}    
\item $\mathbf{(2)}$ $\tau_m$ satisfies (\ref{t2}) iff  $u_m=\frac{\tau_{m} \tau_{m+M+1}}{\tau_{m+1} \tau_{m+M}}$ satisfies   the U-system 
\begin{equation} \label{2o}
u_{m}u_{m+1} \ldots u_{m+N+M-1} 
=\frac{b u_{m+M} u_{m+M+1} \ldots u_{m+N-1}+a}{u_{m+M} u_{m+M+1} \ldots u_{m+N-1}}. 
\end{equation} 
For $N+M$ even, 
\item $\mathbf{(3)}$ $\tau_m$ satisfies (\ref{t1}) iff 
 $u_m=\frac{\tau_{m} \tau_{m+N+2}}{\tau_{m+2} \tau_{m+N}}$ satisfies 
  the U-system
\begin{equation} \label{1e}
u_m u_{m+2} \ldots u_{m+N+M-2}=b+ a u_{m+M} u_{m+M+2} \ldots u_{m+N-2}, 
\end{equation}    
\item  $\mathbf{(4)}$ $\tau_m$ satisfies  (\ref{t2}) iff  $u_m=\frac{\tau_{m} \tau_{m+M+2}}{\tau_{m+2} \tau_{m+M}}$ satisfies  
 the U-system 
 \begin{equation} \label{2e}
u_m u_{m+2} \ldots u_{m+N+M-2} 
=\frac{b u_{m+M} u_{m+M+2} \ldots u_{m+N-2} + a}
{u_{m+M} u_{m+M+2} \ldots u_{m+N-2}}.  
\end{equation} 
\end{proposition}

\begin{proof}
The results in $\mathbf{(1)}$-$\mathbf{(4)}$  
above follow directly by substituting 
the appropriate formula for $u_m$  into \eqref{1o}, \eqref{2o}, \eqref{1e} and \eqref{2e} respectively.  
\end{proof}

Each of the T-systems defines a birational map (cluster map)  in 
dimension $D_j$, 
$$ 
\varphi_j: \quad  \mathbb{C}^{D_j} \rightarrow \mathbb{C}^{D_j},  
\qquad j=1,2,3,4,$$
with dimensions $D_1=D_3=2N+M$,  
$D_2=D_4=2M+N$, 
and each of the reductions $\mathbf{(1)}$-$\mathbf{(4)}$   
described in Proposition  
\ref{redsymp} defines 
a rational 
map  that projects to a lower dimensional space, 
of even dimension $2d_j$, 
that is   
$$ 
\pi_j: \quad \mathbb{C}^{D_j}\rightarrow \mathbb{C}^{2d_j}, 
\qquad j=1,2,3,4, 
$$
where after projection the dimensions are  
$$  
2d_1=2d_2 = M+N-1,  \quad 2d_3=2d_4 = M+N-2.
$$
The U-system recurrences (\ref{1o})--(\ref{2e}) define four 
birational maps in the corresponding reduced space, which we denote by 
$$
\hat{\varphi}_j: \quad  \mathbb{C}^{2d_j} \rightarrow \mathbb{C}^{2d_j},  
\qquad j=1,2,3,4, $$
with the intertwining condition
$$
\pi_j\circ \varphi_j=\hat{\varphi}_j\circ\pi_j
$$
in each case. 
To see how the U-systems and the symplectic structure for the 
corresponding maps $\hat{\varphi}_j$ arise naturally from (\ref{t1}) or (\ref{t2}), 
it is necessary to consider the cluster algebras associated with the latter.   

In order to 
give a 
brief outline of the connection with cluster algebras, 
we will just use 
(\ref{t2}) in dimension $D=2M+N$, since the discussion for (\ref{t1}) is almost identical, and 
follow the approach of \cite{FM}, where 
it was explained in detail how such Somos-type recurrences arise 
from sequences of cluster mutations.  
An initial cluster is defined by the $D$-tuple of variables 
$(\tau_1,\tau_2,\ldots,\tau_{D})$; we can regard the coefficients 
$a,b$ as additional frozen variables, 
which do not mutate. Setting $m=1$ in (\ref{t2}), 
we rewrite an iteration of the T-system in the form of a mutation $\mu_1$, namely   
\begin{equation}\label{mut} 
\mu_1: \quad \tau_{1}'\tau_{1}= a \,\tau_{2 M+1}^{-B_{1,2M+1}}\tau_{N+1}^{-B_{1,N+1}}+ b \,\tau_{M+1}^{B_{1,M+1}}\tau_{N+M+1}^{B_{1,N+M+1}} , 
\end{equation} 
where 
$$ 
B_{1,M+1}=B_{1,N+M+1}=1, \quad B_{1,2M+1}=B_{1,N+1}=-1, 
$$
and  
$$B_{1,k}=0 \qquad\mathrm{for} \quad k\neq M+1,N+1,2M+1,N+M+1 
$$ 
defines the first row of a skew-symmetric integer matrix 
(exchange matrix) $B=(B_{i,k})$ of size $D=2M+N$. 
The entries of the other rows can be found recursively, 
since they are required to  satisfy the formulae 
\beq\label{cond}
B_{i,D}=B_{1,i+1}, 
\qquad 
B_{i+1,k+1}=B_{i,k} + B_{1,i+1}[-B_{1,k+1}]_+ - B_{1,k+1}[-B_{1,i+1}]_+, 
\eeq 
for $i,k \in [1,D-1]$,   
with 
the notation $[x]_+=\max(x,0)$ for real $x$.
For example, setting $M=3$, $N=4$ 
gives $D=10$ and 
\beq\label{bex}
B=\left(
\begin{array}{cccccccccc}
0 & 0 & 0 & 1 & -1 & 0  &  -1 & 1 & 0 & 0\\
0 & 0 & 0 & 0 &  1  & -1 &  0 &-1 & 1 & 0\\
0 & 0 & 0 & 0 &  0 &   1 & -1 &  0 &-1& 1\\
-1&0&0&0&1&0&2&-1&0&-1\\
1&-1&0&-1&0&1&0&1&-1&0\\
0&1&-1&0&-1& 0&1&0&1&-1\\
1&0&1&-2&0&-1&0&0&0&1\\
-1&1&0&1&-1&0&0&0&0&0\\
0&-1&1&0&1&-1&0&0&0&0\\
0&0&-1&1&0&1&-1&0&0&0
\end{array}
\right).
\eeq

The mutation (\ref{mut}) generates a new cluster 
$(\tau_1',\tau_2,\ldots,\tau_{D})$, which only differs 
from the initial one 
in the first component. There is a corresponding 
mutated exchange matrix $B'=\mu_1(B)$, where in general 
the action of the $j$th mutation on $B$ produces 
$B'=\mu_j(B)=(B'_{i,k})$ with entries given by 
$$ 
B'_{i,k} = \begin{cases} 
-B_{i,k} \qquad 
\mathrm{if} \,\, i=j \,\, \mathrm{or}\,\,k=j, 
\\  
B_{i,k}+\frac{1}{2}(|B_{i,j}|B_{j,k}+ B_{i,j}|B_{j,k}|) \quad 
\mathrm{otherwise}, 
\end{cases} 
$$ 
and there is an exchange relation analogous to 
(\ref{mut}) 
describing the action of a general mutation 
$\mu_j$ on 
a cluster, but we omit the details. 

The matrix $B$ defines a quiver (that is, a directed graph) 
without 1-cycles or 2-cycles, and using the indices 
$1,2,\ldots,N$ to label the vertices of the quiver, for each   
$k$ there is an associated quiver mutation at vertex $j$, also 
denoted    $\mu_j$.   
The matrices $B$ being considered here have a particularly special form,
due to the conditions (\ref{cond}), 
which  ensure that 
the action of $\mu_1$ on the exchange matrix corresponds to a cyclic permutation 
of the indices $1,2,\ldots,D$,  and this implies  that, in terms of the cluster variables, applying the sequence of successive mutations 
$\mu_1,\mu_2,\mu_3,\ldots$ etc.\ in order is equivalent to iterating the 
recurrence   (\ref{t2}). In the terminology of \cite{FM}, 
$B$ is said to be cluster mutation-periodic with period 1. 

It is known that, for any skew-symmetric integer matrix $B$, the 
corresponding log-canonical presymplectic form 
\beq\label{pres}
\omega=\sum_{i<k} \frac{B_{i,k}}{\tau_i\tau_k} \, \rd \tau_i\wedge \rd \tau_k 
\eeq 
transforms covariantly under cluster mutations \cite{gsv}, and for the 
particular case at hand more is true: 
this two-form is 
invariant under iteration of the T-system (\ref{t2}), as was proved in \cite{FH} 
for the general case of T-systems (cluster maps) obtained from 
cluster mutation-periodic quivers with period 1. 
The matrix $B$ has even rank $2d$, and by choosing a suitable basis  
${\bf w}_1,\ldots, {\bf w}_{2d}$ for im$\,B$  one can 
construct a projection to 
reduced variables given by Laurent monomials in the initial cluster, that is  
\beq\label{uvars} 
\pi: \quad u_m = {\boldsymbol\tau}^{{\bf w}_m}, \qquad m=1,\ldots,2d
\eeq
(where any integer vector ${\bf a}=(a_j)$ defines a Laurent monomial 
${\boldsymbol\tau}^{{\bf a}}=\prod_j \tau_j^{a_j}$), such that 
the T-system reduces to a symplectic map in terms of the reduced variables. 
Furthermore, in \cite{honeinoue} it was further proved that (up to an overall 
sign)  
there is a unique choice of integer basis for   im$\,B$, 
called a palindromic basis, such that the symplectic map in the reduced 
variables takes the form 
\beq \label{umap} 
\hat{\varphi}: \qquad \Big(u_1,\ldots, u_{2d-1},u_{2d}\Big) 
\mapsto \Big(u_2,\ldots, u_{2d},(u_1)^{-1}{\cal F}\Big), 
\eeq
for a certain rational function 
$
{\cal F} = {\cal F} (u_2,\ldots,u_{2d})
$. 
The birational map (\ref{umap}) defines the U-system associated with the T-system that is specified by the matrix $B$.
It preserves a symplectic form $\hat\omega$ which is log-canonical in the coordinates $(u_i)$, and pulls back to 
the presymplectic form corresponding to $B$, so that 
\beq\label{osymp}
{\hat\omega}=\sum_{i<j}\frac{\hat{B}_{i,j}}{u_iu_j}\rd u_i \wedge\rd u_j, \qquad \hat{\varphi}^*\hat{\omega}=\hat{\omega}, \qquad 
\pi^*\hat{\omega}=\omega,
\eeq 
for a constant skew-symmetric matrix $\hat{B}=(\hat{B}_{i,j})$.

In the particular example (\ref{bex}) above,  $B$ has rank 6, and 
the palindromic basis 
of  im$\,B$, unique up to sign, is given by shifting the entries of 
$$
{\bf w}_1=(1,-1,0,-1,1,0,0,0,0,0)^T, 
$$ 
so that 
$$
{\bf w}_2=(0, 1,-1,0,-1,1,0,0,0,0)^T,\, 
\ldots,  
{\bf w}_6=(0,0,0,0, 0,1,-1,0,-1,1)^T, 
$$ 
and the map $\pi$ in  (\ref{uvars}) coincides with the formula for $u_m$ 
in part  $\mathbf{(2)}$ of Proposition \ref{redsymp}, i.e. $\pi=\pi_2$ in this case. The U-system corresponds to a 6-dimensional map $\hat{\varphi}=\hat{\varphi}_2$, that is 
$$ 
\hat{\varphi}_2: \qquad \Big(u_1,\ldots, u_{5},u_{6}\Big) 
\mapsto \left(u_2,\ldots, u_{6},\frac{a+bu_4}{u_1u_2u_3(u_4)^2u_5u_6}\right), 
$$ 
which is symplectic with respect to the nondegenerate 2-form 
 $\hat{\omega}=\hat{\omega}_2$ defined by (\ref{osymp}) with 
$$
\hat{B} = 
\left( 
\begin{array}{cccccc}
0 & 0 & 0 & 1 & 0 & 0 \\ 
0 & 0 & 0 & 1 & 1 & 0 \\ 
0 & 0 & 0 & 1 & 1 & 1 \\
-1 & -1 & -1 & 0 & 0 & 0  \\ 
0 & -1 & -1 & 0 & 0 & 0 \\ 
0 & 0 & -1 & 0 & 0 & 0
\end{array}
\right). 
$$

The following result shows that, with appropriate assumptions on $M$ and 
$N$, the properties of the preceding example generalize to all of the U-systems 
in Proposition  \ref{redsymp}.                                                                                                                                                                                                                                                                                                                                                                                                                                                                                                                                                                                                                                                                                                                                                                                                                                                                                                                                                                                                                                                                                                                                                                                                                                                                                                                                                                                                                                                                                                                                                                                                                                                                                                                                                                                                                                                                                                                                                                                                                                                                                                                                                                                                                                                                                                                                                                                                                                                                                                                                                                                                                                                                                                                                                                                                                                                                                                                                                                                                                                                                                                                                                                                                                                                                                                                                                                                                                                                                                                                                                                                                                                                                                                                                                                                                                                                                                                                                                                                                                                                                                                                                                                                                                                                                                                                                                                                                                                                                                                                                                                                                                                                                                                                                                                                                                                                                                                                                                                                                                                                                                                                                                                                                                                                                                                                                                                                                                                                                                                                                                                                                                                                                                                                                                                                                                                                                                                                                                                                                                                                                                                                                                                                                                                                                                                                                                                                                                                                                                                                                                                                                                                                                                                                                                                                                                                                                                                                                                                                                                                                                                                                                                                                                                                                                                                                                                                                                                                                                                                                                                                                                                                                                                                                                                                                                                                                                                                                                                                                                                                                                                                                                                                                                                                                                                                                                                                                                                                                                                                 
\begin{theorem}\label{dimthm} 
For coprime $N>M$, 
 each of the U-systems (\ref{1o}) and (\ref{2o}) 
preserves a log-canonical symplectic form in dimension $M+N-1$ 
when $M+N$ is odd, and each of 
the U-systems (\ref{1e}) and (\ref{2e}) 
preserves a log-canonical symplectic form in dimension $M+N-2$ 
when $M+N$ is even.
\end{theorem}                                                                                                                                                                                                                                                                                                                                                                                                                                                                                                                                                                                                                                                                                                                                                                                                                                                                                                                                                                                                                                                                                                                                                                                                                                                                                                                                                                                                                                                                                                                                                                                                                                                                                                                                                                                                                                                                                                                                                                                                                                                                                                                                                                                                                                                                                                                                                                                                                                                                                                                                                                                                                                                                                                                                                                                                                                                                                                                                                                                                                                                                                                                                                                                                                                                                                                                                                                                                                                                                                                                                                                                                                                                                                                                                                                                                                                                                                                                                                                                                                                                                                                                                                                                                                                                                                                                                                                                                                                                                                                                                                                                                                                                                                                                                                                                                                                                                                                                                                                                                                                                                                                                                                                                                                                                                                                                                                                                                                                                                                                                                                                                                                                                                                                                                                                                                                                                                                                                                                                                                                                                                                                                                                                                                                                                                                                                                                                                                                                                                                                                                                                                                                                                                                                                                                                                                                                                                                                                                                                                                                                                                                                                                                                                                                                                                                                                                                                                                                                                                                                                                                                                                                                                                                                                                                                                                                                                                                                                                                                                                                                                                                                                                                                                                                                                                                                                                                                                                                                                                                                                                                                                                                                                                                                                                                                                                                                                                                                                                                                                                                                                                                                                                                                                                                                                                                                                                                                                                                                                                                                                                                                                                                                                                                                                                                                                                                                                                                                                                                                                                                                                                                                                                                                                                                                                                                                                                                                                                                                                                                                                                                                                                                                                                                                                                                                                                                                                                                                                                                                                                                                                                                                                                                                                                                                                                                                                                                                                                                                                                                                                                                                                                                                                                                                                                                                                                                                                                                                                                                                                                                                                                                                                                                                                                                                                                                                                                                                                                                                                                                                                                                                                                                                                                                                                                                                                  The proof of this result is presented in Appendix \ref{appa}, 
where we also provide an explicit description of the U-system Poisson 
brackets in dimension $2d$, which take the log-canonical form 
\beq\label{ubrackets} 
\{u_i,u_j\}=a_{j-i}u_i u_j 
\eeq 
with suitable constants $a_k =-a_{-k}$ for $0\leq k \leq 2d-1$.

\section{Reductions of Hirota's lattice KdV equation}

As was shown in \cite{HKW}, besides the underlying U-systems, 
there is another class of  
lower-dimensional recurrences associated with the 
T-systems (\ref{t1}) and (\ref{t2}), 
corresponding to travelling wave reductions of Hirota's discrete  KdV equation 
(\ref{KdV}) on a two-dimensional lattice. 
In this context, it is necessary to allow 
coefficients $a,b$ that are periodic in the independent variable $m$.    
(One can also have more general dependence on $m$, which leads to 
equations of discrete Painlev\'e type \cite{bershtein, bershtein2, honeinoue, okubo}.)

The $(N,M)$ travelling wave reduction of (\ref{KdV}) is derived by considering 
solutions that are periodic with respect to simultaneous shifts by $N$ steps and $M$ steps  
in the $k,l$ lattice directions, respectively, that is  
\begin{equation} \label{Vtov}
V_{k+N,l+M}=V_{k,l}\implies 
V_{k,l}=v_m, \qquad m=kM-l N
\end{equation}
This is the discrete analogue of the travelling wave reduction for a partial 
differential equation in $1+1$ dimensions, which reduces a function $V(x,t)$ to a 
function $v=v(z)$ satisfying an ordinary differential equation in a single variable $z=x-ct$.  
In the discrete setting, the ratio $N/M\in\mathbb{Q}$ 
corresponds to the wave speed   $c$,  and from  (\ref{KdV})  
we obtain the following ordinary difference equation in terms of the travelling wave variable $m$: 
\begin{equation} \label{RedkdV}
v_{m+{N}+{M}}-v_m=\alpha\left(\frac{1}{v_{m+N}}-\frac{1}{v_{m+M}}
\right).
\end{equation}
(Equivalently, setting $m=l N-kM$ as  in \cite{HKW} leads to the same equation, 
due to a symmetry of the discrete 
KdV equation.)
An example of one of these travelling wave solutions of the lattice KdV equation is presented as a 3D plot in 
Figure \ref{f1}; another view of the same solution is provided by the contour plot in Figure \ref{f2} (which 
is somewhat reminiscent of a contour plot of a genus 3 solution of the continuous KP equation in \cite{sdeco}). 

The following key observation, from  \cite{HKW}, is that the same discrete reduction (\ref{RedkdV}) 
is related both to the T-system (\ref{t1}) with $a\to -\alpha$, $b\to\beta_m$, and to 
the T-system (\ref{t2}) with  $a\to\alpha$, $b\to\beta_m'$,
where $\beta_m,\beta_m'$ are periodic coefficients with periods $M,N$ respectively.

\begin{proposition} \label{kdvred} 
Suppose that a travelling wave solution 
of the lattice KdV equation is given by 
\begin{equation}\label{vtau} 
v_m=\frac{\tau_m  \tau_{m+N+M}}{\tau_{m+M} \tau_{m+N}}, 
\end{equation}
satisfying the ordinary difference equation \eqref{RedkdV}. Then $\tau_m$ satisfies 
the following two bilinear equations: 
\begin{eqnarray}
\tau_{m+2N+M}\tau_m &=& \beta_m \tau_{m+N+M}  \tau_{m+N} - \alpha \tau_{m+2 N}\tau_{m+M},  \qquad \beta_{m+M}=\beta_m, \label{1HKdV} \\
\tau_{m+2M+N}\tau_{m} &=& \beta'_m \tau_{m+N+M}  \tau_{m+M}+\alpha \tau_{m+2 M}\tau_{m+N},\qquad   \beta'_{m+N}=\beta'_m. \label{2HKdV}
\end{eqnarray}
Conversely, 
if $\tau_m$ is a solution of either \eqref{1HKdV} or \eqref{2HKdV}, then  
$v_m$ given by \eqref{vtau} satisfies \eqref{RedkdV}.
\end{proposition}

The T-systems \eqref{t1} and \eqref{t2} with constant coefficients can be seen as a particular case of (\ref{1HKdV}) 
and (\ref{2HKdV}), where  
$\beta_m=\beta'_m=b$ for any $m$. So it follows that $v_m$ defined by \eqref{vtau} satisfies \eqref{RedkdV} whenever  
$\tau_m$ satisfies one of these discrete Hirota reductions with constant coefficients, but the converse statement is not true.

\subsection{Lax representation and first integrals} \label{subsecLax}

It was shown in \cite{HKW} that any 
bilinear difference equation of the form \eqref{genform}, 
or a suitable generalization with periodic coefficients, admits a Lax pair derived from the 
Lax representation of the discrete Hirota equation. In the cases of (\ref{1HKdV}) and (\ref{2HKdV})  this construction involves
$N \times N$ and $\min (N, 2M ) \times \min (N, 2M )$ Lax matrices respectively. However, in these cases there is also a $2 \times 2$ Lax representation derived from the Lax representation of the lattice KdV equation.

The lattice KdV equation \eqref{KdV} is equivalent to the discrete 
zero curvature equation 
\begin{equation} \label{LaxKdV}
{\bf L}(V_{k,l+1},V_{k+1,l+1}){\bf M}(V_{k,l})={\bf M}(V_{k+1,l}){\bf L}(V_{k,l},V_{k+1,l}),
\end{equation}
where
\begin{equation} \label{LaxPair2}
{\bf L}(V,W)=\left(
\begin{array}{cc}
 V-\frac{\alpha}{W} & \ \lambda \\
 1 & 0
\end{array}
\right), \ {\bf M} (V)=\left(
\begin{array}{cc}
V & \ \lambda \\
 1 & \frac{\alpha}{V}
\end{array}
\right),
\end{equation}
and  $\lambda$ is a spectral parameter.

It is well known that the Lax representation of quadrilateral lattice equations gives rise to Lax representations of their periodic reductions (see 
e.g.\ \cite{staircase} and references therein). 
First integrals of these systems are derived from the 
spectrum of an associated monodromy 
matrix. 
In the case of the lattice KdV equation, from the periodic reduction \eqref{Vtov} and the Lax representation \eqref{LaxKdV} 
we derive the reduced version of the discrete zero curvature Lax representation of  \eqref{RedkdV}, 
that is 
\begin{equation} \label{laxvkdv}
{\bf L}(v_{m},v_{m+M}){\bf M}(v_{m+N})={\bf M}(v_{m+N+M}){\bf L}(v_{m+N},v_{m+N+M}). 
\end{equation}
Hence,  
by making the substitution \eqref{vtau} in \eqref{laxvkdv},
a $2 \times 2$ Lax representation can be obtained for the $(N,M)$ periodic reduction 
\eqref{RedkdV}, and consequently  
for the 
discrete bilinear 
equations (\ref{1HKdV}) and (\ref{2HKdV}), as well as  
for the corresponding U-systems of Proposition \ref{redsymp}. 

For coprime  $N, M$ with $N>M$,   
for convenience we write  
$$ 
{\bf M}_j ={\bf M}(v_j), \qquad {\bf L}_j =  {\bf {L}}(v_j,v_{j+M}), 
$$ and 
define the 
monodromy matrix to be 
\beq\label{Monodr1}
\mathcal{M}_m=\prod_{i=0}^{M-1} {\bf {M}}_{m+r_i+N}{\bf {L}}_{m+r_i+N-M}
{\bf {L}}_{m+r_i+N-2M}  {\bf {L}}_{m+r_i+N-3M} 
\cdots{\bf{L}}_{m+r_{i+1}},
\eeq 
where 
\begin{equation*} \label{ri}
r_k=kN \bmod M 
\end{equation*}
(the product in \eqref{Monodr1} is arranged from left to right).  For $M>1$, 
we also consider the matrix 
\begin{align*} \label{Lmatrix}
\mathcal{L}_m:=&{\bf {L}}_{m}^{-1}{\bf {M}}_{m+N+M}{\bf {L}}_{m+N} 
{\bf {L}}_{m+N-M} {\bf {L}}_{m+N-2M}  
\cdots{\bf{L}}_{m+r_{1}} \nonumber \\
=& {\bf {M}}_{m+N}{\bf {L}}_{m+N-M} {\bf {L}}_{m+N-2M}    
 \cdots{\bf{L}}_{m+r_{1}},
\end{align*}
where in the last equality we used \eqref{laxvkdv}. Now, from the above definitions, a direct computation 
shows that $ \mathcal{M}_m$ satisfies the discrete Lax equation 
\begin{equation} \label{monodromyLax}
\mathcal{M}_m \mathcal{L}_m= \mathcal{L}_m \mathcal{M}_{m+1}. 
\end{equation}
Therefore, the following corollary holds. 
\begin{corollary}
For coprime  $N, M$ with $N>M$, the $(N,M)$ KdV periodic reduction \eqref{RedkdV} preserves the spectrum of the monodromy matrix \eqref{Monodr1}. 
\end{corollary}
Equivalently, the birational map  corresponding  to the KdV recurrence \eqref{RedkdV}, that is  
\begin{equation*} \label{KdVmap}
\phi: \quad (v_0,v_1,\dots,v_{N+M-1})\mapsto\left(v_1,v_2,\dots,
v_0+\alpha\Big(\frac{1}{v_{N}}-\frac{1}{v_{M}}\Big)\right), 
\end{equation*}
preserves the spectral curve 
\beq\label{spectral}
\det ( \mathcal{M}(\lambda) - \nu \, \mathbf{1}) = 0  
\eeq  for the monodromy matrix $\mathcal{M}(\lambda)= \mathcal{M}_0$ 
obtained by setting $m=0$ in (\ref{Monodr1}), namely 
\begin{equation} \label{monodr1}
\mathcal{M}(\lambda) =\prod_{i=0}^{M-1}{\bf {M}}_{r_i+N}{\bf {L}}_{r_i+N-M} 
{\bf {L}}_{r_i+N-2M}\cdots 
{\bf{L}}_{r_{i+1}}, 
\end{equation} 
where the dependence of ${\bf {L}}_j$ and ${\bf {M}}_j$ 
on the spectral parameter $\lambda$ is implicit from (\ref{LaxPair2}).


\section{The odd case}

From the U-systems described in Proposition \ref{redsymp}, it is evident that the two different cases of 
odd/even $N+M$   
are structurally different. 
Thus we continue our analysis by considering the odd case first. 
The periodic coefficients that appear in the discrete 
bilinear equations \eqref{1HKdV} and \eqref{2HKdV} 
introduce periodic coefficients in the 
corresponding U-systems of Proposition \ref{redsymp}. Furthermore, the $u$-variables of the U-systems are related with the $v$-variables of 
KdV periodic reductions. 

\begin{proposition} \label{changevarodd}
Let $N+M$ be odd. If $v_m,u_m,u'_m$ are related by 
\begin{equation} \label{odvarall}
v_m = u_m u_{m+1} \dots u_{m+M-1}=u'_m u'_{m+1} \dots u'_{m+N-1}, 
\end{equation}
then the following statements are equivalent: 
\begin{enumerate}
\item[(i)] $v_m$ satisfies the $(N,M)$ KdV periodic reduction \eqref{RedkdV};  
\item[(ii)] $u_m$ satisfies the U-system (\ref{1o}) 
with periodic coefficients, that is 
\begin{equation} \label{1U}
u_{m}u_{m+1} \ldots u_{m+N+M-1}=\beta_m - \alpha u_{m+M} u_{m+M+1} \ldots u_{m+N-1}, \  \ \beta_{m+M}=\beta_m ;
\end{equation}    
\item[(iii)] $u'_m$ satisfies the U-system (\ref{2o}) 
with periodic coefficients, that is 
\begin{equation} \label{2U}
u'_{m}u'_{m+1} \ldots u'_{m+N+M-1}=\beta'_m+\frac{\alpha}{u'_{m+M} u'_{m+M+1} \ldots u'_{m+N-1}}, \ \   \beta'_{m+N}=\beta'_m .
\end{equation} 
\end{enumerate}

\end{proposition}

\begin{proof}
According to Prop. \ref{kdvred}, $v_m=\frac{\tau_m  \tau_{m+N+M}}{\tau_{m+M} \tau_{m+N}}$ satisfies \eqref{RedkdV} if and only 
$\tau_m$ satisfies \eqref{1HKdV} and  from Prop. \ref{redsymp}, for any $\beta_m$, $\tau_m$ satisfies \eqref{1HKdV} if and only if 
$u_m=\frac{\tau_{m} \tau_{m+N+1}}{\tau_{m+1} \tau_{m+N}}$ satisfies   
\begin{equation*}
u_{m}u_{m+1} \ldots u_{m+N+M-1}=\beta_m - \alpha u_{m+M} u_{m+M+1} \ldots u_{m+N-1}. 
\end{equation*}    
Therefore, $u_m$ satisfies \eqref{1U}, if and only if 
$$v_m=\frac{\tau_m  \tau_{m+N+M}}{\tau_{m+M} \tau_{m+N}}=
\frac{\tau_{m} \tau_{m+N+1}}{\tau_{m+1} \tau_{m+N}} \frac{\tau_{m+1} \tau_{m+N+2}}{\tau_{m+2} \tau_{m+N+1}} \dots 
\frac{\tau_{m+M-1} \tau_{m+N+M}}{\tau_{m+M} \tau_{m+N+M-1}}=u_m u_{m+1} \dots u_{m+M-1}$$
satisfies \eqref{1U}. 

In a similar way, we derive that \eqref{RedkdV} is equivalent to \eqref{2U} 
for 
$$v_m=\frac{\tau_m  \tau_{m+N+M}}{\tau_{m+M} \tau_{m+N}}=
\frac{\tau_{m} \tau_{m+M+1}}{\tau_{m+1} \tau_{m+M}} \frac{\tau_{m+1} \tau_{m+M+2}}{\tau_{m+2} \tau_{m+M+1}} \dots 
\frac{\tau_{m+N-1} \tau_{m+N+M}}{\tau_{m+N} \tau_{m+N+M-1}}=u'_m u'_{m+1} \dots u'_{m+N-1}, $$ 
from the second $U$-system \eqref{2o} corresponding to \eqref{2HKdV}.
\end{proof}

Using the substitution \eqref{odvarall}, the Lax representation of the $(N,M)$ periodic reduction of lattice KdV 
gives rise to 
a Lax representation of the two U-systems (\ref{1U}-\ref{2U}) and the corresponding monodromy matrix \eqref{monodr1} generates first integrals 
of the U-systems.


\subsection{Bi-Poisson structure of the lattice KdV periodic reductions} 
We have seen that two different  bilinear equations with periodic coefficients, obtained as reductions 
of the discrete Hirota equation,  give rise to the same periodic reduction of the lattice KdV equation. 
By Theorem \ref{dimthm}, the associated U-systems inherit a nondegenerate log-canonical Poisson structure from the  mutation periodic quiver corresponding 
to each of the bilinear equations. 
In the case of coprime $N, M$, with $N+M$ odd, the Poisson structure of the two U-systems gives rise to two Poisson structures of the 
corresponding discrete KdV 
reductions. We will prove that these structures are compatible in the sense that any linear combination of them also defines a Poisson bracket. 
This fact will be the key to demonstrating the integrability of the lattice KdV reductions, and consequently of the original U-systems.
The main result is described in the next theorem.

\begin{theorem} \label{12bras}
Let $N, M$ be coprime with $N>M>1$ and $N+M$ odd. 
The brackets 
\begin{eqnarray} \label{poisbrack1}
\{ v_i, v_j \}_1 &=& \begin{cases}
c_{j-i} v_i v_j,  &j-i \neq  N, \\ 
c_{j-i} v_i v_j+ c_N \alpha, &j-i = N, 
\end{cases} 
\end{eqnarray} 
\begin{eqnarray} \label{poisbrack2}
\{ v_i, v_j \}_2 &=& \begin{cases}
d_{j-i} v_i v_j,  &j-i \neq kM, \\ 
d_{j-i} v_i v_j +d_M  (-\alpha)^k \prod_{l=1}^{k-1} v_{i + lM}^{-2}, &j-i = k M, 
\end{cases}
\end{eqnarray}
for $0 \leq i<j \leq N+M-1$,
where (up to rescaling by an arbitrary constant) 
\begin{equation} \label{coefcd}
c_k = d_k = (-1)^{h_k}, 
\end{equation}
with \begin{equation} \label{hmod}   
h_k=\frac{k}{M} \bmod (N+M) 
\end{equation}
for $k=1,\dots, N+M-1$,
define two compatible Poisson structures on $\mathbb{C}^{N+M}$  of rank $N+M-1$. 
Furthermore, the map (\ref{KdVmap}) corresponding to  the $(N,M)$ 
reduction of the lattice KdV equation 
is a Poisson map with respect to both 
of these brackets.
\end{theorem}
We will devote the rest of this section in the proof of this theorem. 

\begin{figure} \centering 
\includegraphics[width=7.33cm,height=7.33cm,keepaspectratio]{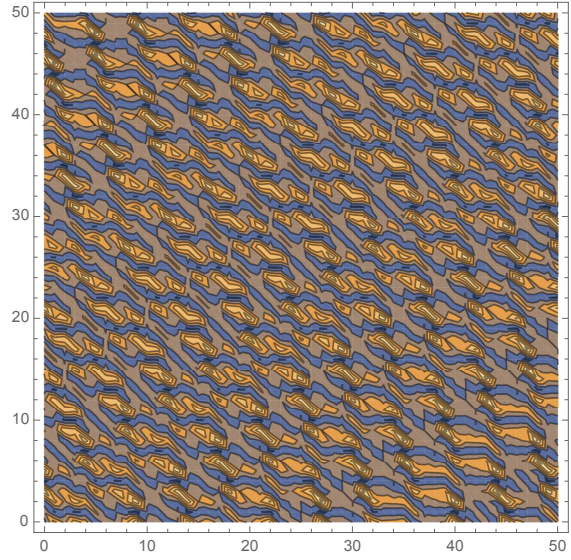}
\caption{\small{Contour plot of the same 3-phase travelling wave solution of the discrete KdV equation 
(\ref{KdV}) as in Figure \ref{f1}.  }}
\label{f2}
\end{figure}

\subsubsection{First Poisson bracket}
For coprime $N,M$ with $N>M$, $N+M$ odd, we consider the U-system  (\ref{1U}) 
 that corresponds to the discrete Hirota reduction \eqref{1HKdV}. Each iteration of the associated  map  
$\hat{\varphi}_1:\, \mathbb{C}^{M+N-1} \rightarrow \mathbb{C}^{M+N-1}$ defined by  
\begin{equation*} 
\hat{\varphi}_1: \quad (u_0,u_1,\dots, u_{N+M-3},u_{N+M-2}) 
\mapsto \left(
u_1,u_2,\dots,u_{N+M-2},\frac{\beta_m- \alpha u_{M} u_{M+1} \ldots u_{N-1} }{u_0 u_1 \dots u_{N+M-2}}\right)  
\end{equation*} 
is symplectic (really it is a family of maps depending on $m$, with the parameter $\beta_m$ varying with period $M$). 
Equivalently, $\hat{\varphi}_1$  a Poisson map with respect to the 
nondegenerate log-canonical bracket (\ref{ubrackets}), which we will 
denote by $\{,\}_u$. 
According to Prop. \ref{changevarodd}, $v_m=u_m  u_{m+1}\dots u_{m+M-1}$ satisfies  
the $(N,M)$ periodic reduction of the lattice KdV equation
\eqref{RedkdV}, which corresponds to the birational map
\beq\label{kdvphi} 
\phi: \quad (v_0,v_1,\dots,v_{N+M-1})\mapsto \left(v_1,v_2,\dots,v_0 
+\alpha\Big(\frac{1}{v_{N}}-\frac{1}{v_{M}}\Big)\right).
\eeq 
We can write the variables $v_0,v_1,\dots,v_{N+M-1}$ of this map 
in terms of 
$u_0,\dots ,u_{N+M-2}$ using the recurrence \eqref{1U}, as
$$v_m=u_m u_{m+1} \dots u_{m+M-1}  \  \text{for} \  0 \leq  m \leq N-1, $$
and 
$$v_{N+k}=\frac{\beta_{k}-\alpha u_{k+M} u_{k+M+1} \dots u_{k+N-1}}{u_{k} u_{k+1} \dots u_{k+N-1}}  \  \text{for} \  0 \leq  k \leq M-1,$$
or equivalently, by setting $u_k u_{k+1} \dots u_{k+M-1}=v_k$,  
\begin{eqnarray}
v_m = \begin{cases} \label{vfromu1}
u_m u_{m+1} \dots u_{m+M-1},  & 0 \leq  m \leq N-1 \\ 
\frac{\beta_{m-N}}{u_{m-N} u_{m-N+1} \dots u_{m-1}}- \frac{\alpha}{v_{m-N}},  & N \leq  m \leq N+M-1.
\end{cases}
\end{eqnarray}

Next, we 
evaluate the Poisson brackets $\{ v_0, v_m \}_u$, $0 < m \leq N+M-1$.  
For $0< m \leq N-1$, we have 
\begin{eqnarray*}
\{ v_0, v_m \}_u=\sum_{i,j=0}^{N+M-2}a_{j-i}u_i u_j \frac{\partial{v_0}}{\partial{u_i}} \frac{\partial{v_m}}{\partial{u_j}}
=v_0 v_m \left(\sum_{i=0}^{M-1} \sum_{j=m}^{m+M-1} a_{j-i}\right).
\end{eqnarray*}
Then for $N\leq m \leq N+M-1$, we find 
\begin{eqnarray*}
\{ v_0, v_m \}_u &=&\{ v_0, \beta_{m-N}(u_{m-N} u_{m-N+1} \dots u_{m-1})^{-1} \}_u
+ \frac{\alpha}{v_{m-N}^2} \{ v_0, v_{m-N} \}_u \\ 
&=& -\frac{\beta_{m-N} v_0}{u_{m-N} u_{m-N+1} \dots u_{m-1}} 
\left(\sum_{i=0}^{M-1} \sum_{j=m-N}^{m-1} a_{j-i}\right) + \frac{\alpha}{v_{m-N}^2} \{ v_0, v_{m-N} \}_u \\ 
&=& \Big(-v_m- \frac{\alpha}{v_{m-N}}\Big)
\left(\sum_{i=0}^{M-1} \sum_{j=m-N}^{m-1} a_{j-i}\right) + \frac{\alpha}{v_{m-N}^2} \{ v_0, v_{m-N} \}_u .
\end{eqnarray*}
So, for any $0 \leq m \leq N+M-1$ we can evaluate the Poisson bracket of $\{ v_0, v_m \}_u$ in terms of  
$v_0,v_1,\dots, v_{N+M-1}$ by using the recurrence   
\begin{eqnarray} \label{recP1}
\{ v_0, v_m \}_u = \begin{cases}
c_m v_0 v_m,  & 0 \leq m \leq N-1,  \\ 
c_m v_0 v_m+ c_m \alpha \frac{v_0}{v_{m-N}}+\frac{\alpha}{v_{m-N}^2} \{ v_0,v_{m-N} \}_u, &N \leq  m \leq N+M-1,
\end{cases}
\end{eqnarray}
where 
\begin{eqnarray} \label{coefcm}
c_m= \begin{cases}
  \sum\limits_{i=0}^{M-1} \sum \limits_{j=m}^{m+M-1} a_{j-i},  & 0 <  m \leq N-1,  \\ 
- \sum \limits_{i=0}^{M-1} \sum \limits_{j=m-N}^{m-1} a_{j-i},   & N \leq  m \leq N+M-1. 
\end{cases}
\end{eqnarray}
Additionally, we define $c_{-m}=-c_m$, for $0\leq m \leq N+M-1$. 

Now, by considering the Poisson property of the map $\hat{\varphi}_1$ we can prove the following lemma.

\begin{lemma} \label{lemmacoefmap1}
For $0<m \leq N+M-1$, the coefficients $c_m$ defined by  \eqref{coefcm} 
with $c_m=-c_{-m}$ satisfy the equations
\begin{equation}\label{cmeq}
c_m=-c_{N+M-m}=-c_{m-N}=-c_{m-M}.
\end{equation}
\end{lemma}
The proof appears in Appendix \ref{appendix1}. 
This lemma leads to a closed-form expression for the Poisson brackets $\{ v_i, v_j \}_u$ purely 
 in terms of the variables  
$v_0,v_1,\dots, v_{N+M-1}$, determined by the coefficients $c_m$, 
so that the bracket $\{,\}_u$ lifts to a bracket on $\mathbb{C}^{N+M}$ for the $v_i$, denoted 
$\{,\}_1$ and  
 given by
(\ref{poisbrack1}).

To see how this comes about, 
note that 
for $N<m < N+M$, we have $0<m-N <M<N$, so from \eqref{recP1} 
$$\{ v_0, v_m \}_u=c_m v_0 v_m+ c_m \alpha \frac{v_0}{v_{m-N}}+c_{m-N} \alpha \frac{v_0}{v_{m-N}}$$
and from Lemma \ref{lemmacoefmap1}, $\{ v_0, v_m \}_u=c_m v_0 v_m$. Furthermore, 
$\{ v_0, v_N \}_u=c_N v_0 v_N+c_N \alpha$ and $\{ v_0, v_m \}_u=c_m v_0 v_m$, for $0<m <N$. So,  since 
$\hat{\varphi}_1$ is a Poisson map and $\hat{\varphi}_1^* v_i=v_{i+1}$, the bracket  $\{,\}_u$ lifts to a bracket 
(\ref{poisbrack1}) of the form for the $v_i$. Moreover, by construction (\ref{KdVmap}) is a Poisson map with 
respect to $\{,\}_1$, which proves the first part of Theorem \ref{12bras}, except that it remains to show that 
the coefficients $c_k$ are given by (\ref{coefcd}), up to rescaling by an arbitrary constant. In due course we shall 
see that the latter follows from (\ref{cmeq}).

\subsubsection{Second Poisson bracket}

Now, for 
coprime $N>M$  with $N+M$ odd, 
we consider the second U-system  (\ref{2U}), 
with periodic coefficients $\beta'_{m+N}=\beta'_m$, that corresponds to the discrete Hirota reduction \eqref{2HKdV}.
The map  
$\hat{\varphi}_2:\,\mathbb{C}^{M+N-1} \rightarrow \mathbb{C}^{M+N-1}$, given by 
\begin{equation*} 
\hat{\varphi}_2: \quad (u'_0,\dots,u'_{N+M-2})\mapsto 
\left(u'_1,\dots,u'_{N+M-2},\frac{\beta'_m u'_{M} u'_{M+1} \ldots u'_{N-1}+\alpha}
{(u'_{0}u'_{1} \ldots u'_{N+M-2})(u'_{M} u'_{M+1} \ldots u'_{N-1})}\right)
\end{equation*} 
is symplectic, with the associated nondegenerate log-canonical Poisson bracket of the form 
\eqref{ubrackets}.  It turns out 
that the coefficients of this bracket 
for $\hat{\varphi}_2$ are the same as for (this is proved  in Appendix A), but to distinguish between the 
coordinates for the  two 
different U-systems we denote the bracket for $\hat{\varphi}_2$ by 
$\{ \ , \ \}_{u'}$. 
In this case, the quantities $v_m=u'_m  u'_{m+1}\dots u'_{m+N-1}$ satisfy  the KdV periodic reduction \eqref{RedkdV}. 
From the second  
$U$-system we can wite  
$$v_m=u'_m u'_{m+1} \dots u'_{m+N-1},  \  \text{for} \  0 \leq  m \leq M-1,$$
$$v_{M+k}=\frac{\beta'_{k}}{u'_{k} u'_{k+1} \dots u'_{k+M-1}} + \frac{\alpha}{u'_{k} u_{k'+1} \dots u_{k'+N-1}},  \  \text{for} \  0 \leq  k \leq N-1,$$
and by setting $u'_k u'_{k+1} \dots u'_{k+N-1}=v_k$, we derive  
\begin{eqnarray}
v_m = \begin{cases} \label{vfromu2}
u'_m u'_{m+1} \dots u'_{m+N-1},  & 0 \leq  m \leq M-1, \\ 
\frac{\beta'_{m-M}}{u'_{m-M} u'_{m-M+1} \dots u'_{m-1}}+\frac{\alpha}{v_{m-M}},  & M \leq  m \leq N+M-1.
\end{cases}
\end{eqnarray}

As before, we evaluate the Poisson brackets $\{ v_0, v_m \}_{u'}$,  for $ 0 < m \leq N+M-1.$  After some calculations 
we arrive at  
\begin{eqnarray} \label{recP2}
\{ v_0, v_m \}_{u'} = \begin{cases}
d_m v_0 v_m,  & 0 <  m \leq M-1,  \\ 
d_m v_0 v_m- d_m \alpha \frac{v_0}{v_{m-M}}-\frac{\alpha}{v_{m-M}^2} \{ v_0,v_{m-M} \}_{u'}, &M \leq  m \leq N+M-1,
\end{cases}
\end{eqnarray}
where 
\begin{eqnarray} \label{coefdm}
d_m= \begin{cases}
  \sum\limits_{i=0}^{N-1} \sum \limits_{j=m}^{m+N-1} a_{j-i},  & 0 <  m \leq M-1,  \\ 
- \sum \limits_{i=0}^{N-1} \sum \limits_{j=m-M}^{m-1} a_{j-i},   & M \leq  m \leq N+M-1. 
\end{cases}
\end{eqnarray}
Additionally, we define $d_{-m}=-d_m$, for $0\leq m \leq N+M-1$.

\begin{lemma} \label{lemmacoefmap2}
For $0<m \leq N+M-1$, the coefficients $d_m$  defined by \eqref{coefdm} with $d_m= -d_{-m}$ satisfy the equations 
as for $c_m$ in Lemma \ref{lemmacoefmap1}, that is 
\begin{equation}\label{dmeq}
d_m=-d_{N+M-m}=-d_{m-N}=-d_{m-M}.
\end{equation}
\end{lemma}
The proof of this lemma appears in Appendix \ref{appendix2}.  

Similarly to the 
result of Lemma \ref{lemmacoefmap1}, the latter result allows the brackets $\{v_i,v_j\}_{u'}$ to be written 
as closed form expressions in terms of $v_i$ only, determined by the coefficients $d_m$,  so that the bracket $\{,\}_{u'}$ lifts to a bracket 
on $\mathbb{C}^{N+M}$, denoted 
$\{,\}_2$ and  
 given by
(\ref{poisbrack2}).

To see this, observe that 
for $M<m < 2M$, we have $0<m-M <M$, so from \eqref{recP2} 
$$\{ v_0, v_m \}_{u'}=d_m v_0 v_m- d_m \alpha \frac{v_0}{v_{m-M}}-d_{m-M} \alpha \frac{v_0}{v_{m-M}}$$
and from lemma \ref{lemmacoefmap2} we derive $\{ v_0, v_m \}_{u'}=d_m v_0 v_m.$ 
Similarly, by induction we can show that  
\begin{equation*} \label{prpp21}
\{ v_0, v_m \}_{u'}=d_m v_0 v_m, 
\end{equation*}
for any $m$ with $0<kM<m < (k+1)M<N+M$.  

Moreover, \eqref{recP2} implies  $$\{ v_0, v_M \}_{u'}=d_M v_0 v_M- d_M \alpha.$$ So we obtain 
$$ \{ v_0, v_{2M} \}_{u'}=d_{2M} v_0 v_{2M}- d_{2M} \alpha \frac{v_0}{v_M}-\frac{\alpha}{v^2_M}(d_M v_0 v_M-d_M \alpha)
= d_{2M} v_0 v_{2M} + d_M \frac{\alpha^2}{v^2_M},$$ 
where in the last equality we used that $d_M=-d_{2M}$ from Lemma \ref{lemmacoefmap2}, and by induction 
we can show that for $k>1$, 
 $$ \{ v_0, v_{kM} \}_{u'}= d_{kM} v_0 v_{k M} +\frac{d_M  (-\alpha)^k }{v_M^2 v_{2M}^2 \dots v_{(k-1)M}^2}.$$ 
Finally, we have obtained explicit expressions for 
$\{ v_0, v_m \}_{u'}$,  in terms of the $v_i$ only, for $0\leq m\leq N+M-1$, and then the fact that   
$\hat{\varphi}_2$ is a Poisson map with $\hat{\varphi}_2^*v_i=v_{i+1}$ yields the required formulae for 
the bracket (\ref{poisbrack2}), and completes the next 
part of the proof of Theorem \ref{12bras}, apart from showing that (up to an overall constant), 
the coefficients $d_k$ must have the form (\ref{coefcd}). Furthermore, it remains to prove that the 
brackets $\{,\}_1$  and $\{,\}_2$ are compatible. These points are addressed in the next subsection.



\subsubsection{Coefficients and compatibility} 

So far we have proved that the map $\phi$ corresponding to a lattice KdV reduction, 
given by (\ref{kdvphi}), is Poisson with  respect to the Poisson brackets $\{ \ , \ \}_1$ and $\{ \ , \ \}_2$ 
in (\ref{poisbrack1}) and (\ref{poisbrack2}),  
 with 
 coefficients  $c_m=-c_{-m}$, $d_m=-d_{-m}$ defined by  \eqref{coefcm} 
and  \eqref{coefdm}, respectively, and these two sets of  coefficients satisfy the same 
conditions, namely \eqref{cmeq}, or equivalently \eqref{dmeq}.
We now show that 
the latter conditions uniquely determine the coefficients in the form 
(\ref{coefcd}) with (\ref{hmod}),  up to an overall constant.

\begin{lemma}
If $c_m$ satisfies the conditions  \eqref{cmeq} and $c_m=-c_{-m}$ for $0<m \leq N+M-1$,   
then $c_k=(-1)^{h_k} c$, for $k=1,\dots, N+M-1$, where $c$ is an arbitrary constant and 
$h_k$ is given by (\ref{hmod}). \end{lemma}

\begin{proof}
We set $c_{M}=-c$, where $c$ is an arbitrary constant. From the conditions \eqref{cmeq}, we have that  
$c_{(i+1) M}=-c_{i M}$. Hence, $c_{i M}=(-1)^{i}c$, for any integer $i \geq1$, such that 
$1\leq i M \leq N+M-1$.  
Furthermore, from  \eqref{cmeq} we derive that 
$c_{iM-(j-1) N}=-c_{i M-j N}$. Therefore,
\begin{equation} \label{auxcm}
c_{i M-j N}= (-1)^{j} c_{i M}=(-1)^{j} (-1)^{i}c=(-1)^{i+j} c,
\end{equation}
for $j\geq 1$, with $1\leq i M-j N \leq N+M-1$.
Now, let us consider an integer $k \in [1,N+M-1]$ and 
$h_k=\frac{k}{M} \bmod (N+M)$. That is 
$M h_k-k=\ell (N+M)$, for some integer $\ell$. So,  $k= (h_k-\ell)M-\ell N$ and from \eqref{auxcm} we conclude that 
$c_k=(-1)^{h_k-\ell+\ell} c=(-1)^{h_k} c$. 
\end{proof} 

The coefficients $d_m$ of the second Poisson satisfy the same conditions (Lemma \ref{lemmacoefmap2}). Hence, 
 $d_k=(-1)^{h_k} c$, for $k=1,\dots, N+M-1$, where $d$ is an arbitrary constant and 
$h_k$ is given by (\ref{hmod}). By choosing $c=d=1$, we derive the Poisson brackets of Theorem \ref{12bras}. 

It remains to show the compatibility of the two Poisson brackets. 
To see this, consider 
$$\{ ., . \}_3 =\{.,.\}_1-\{.,.\}_2.$$
Then, for $0\leq i<j \leq N+M-1$, from  (\ref{poisbrack1}) and (\ref{poisbrack2}) together with  
(\ref{coefcd}) we derive that 
\begin{eqnarray}
\{ v_i, v_j \}_3 &=& \begin{cases}
(-\alpha)^k \prod_{l=1}^{k-1} v_{l M+i}^{-2}, &j-i = k M,  \label{poi3} \\ 
\alpha, &j-i = N, \\ 
0, & otherwise.
\end{cases} 
\end{eqnarray}
This is a Poisson bracket that coincides (under the transformation $v_i \mapsto \frac{1}{v_1}$ and by inserting the parameter $\alpha$) 
with the one that is derived from the Lagrangian structure of the lattice KdV equation 
in \cite{HKQT}.  
Therefore, any linear combination $\lambda_1\{.,.\}_1+\lambda_2\{.,.\}_2$ satisfies the Jacobi identity and so defines a Poisson bracket.

An example of the aforementioned bi-Hamiltonian formalism appears in \cite{HKW}, where the case 
$N=3$, $M=2$  is presented in detail  
and Liouville integrability is proved for the corresponding lattice KdV reduction and U-systems.

\section{Liouville integrability}
In this section we will prove the Liouville integrability of the $(N,M)$ KdV periodic reductions in the case of coprime  $N, M$ with $N>M$ and $N+M$ odd.

\subsection{Monodromy matrix refactorization}
  
 As stated in section \ref{subsecLax}, the integrals of the Poisson map \eqref{KdVmap} are derived from the trace of the monodromy matrix  \eqref{monodr1}.
Let us now consider the matrix 
\begin{equation} \label{laxg}
\tilde{{\bf L}}(g,\lambda)=\left(
\begin{array}{cc}
g & \ \lambda \\
 1 & 0
\end{array}
\right).
\end{equation}
We notice that
$$\mathbf{M}(v_j)\mathbf{L}(v_i,v_j)=\tilde{{\bf L}}(v_j,\lambda-\alpha)\tilde{{\bf L}}(v_i,\lambda)$$
and 
$$\mathbf{L}({v}_i,{v}_j)=\tilde{{\bf L}}(v_{i}-\alpha/v_j,\lambda).$$
Therefore, the monodromy matrix $\mathcal{M}=\mathcal{M}(\lambda )$ in \eqref{monodr1} can be rewritten as 
\small
\begin{equation*} \label{monodr2}
\mathcal{M}=\prod_{i=0}^{M-1}\tilde{{\bf L}}(v_{r_i+N},\lambda-\alpha)\tilde{{\bf L}}(v_{r_i+N-M},\lambda)
\tilde{{\bf L}}(v_{r_i+N-2M}- 
{\alpha}/{v_{r_i+N-M}},\lambda)\cdots \tilde{{\bf L}}(v_{r_{i+1}}-{\alpha}/{v_{r_{i+1}+M}},\lambda).
\end{equation*}
\normalsize
This form of the monodromy matrix motivates us to consider a new set of variables that we present below. 

\subsection{A new set of coordinates}  \label{sectiongvar}

We consider the variables $g_i$, $i=0,\dots v_{N+M-1}$, defined by 
\begin{eqnarray}
g_i &=& \begin{cases}
v_i,    \ \  \ \text{for}  \ N-M \leq i \leq N+M-1, & \label{gvar} \\ 
v_i-\frac{\alpha}{v_{i+M}}, \  \ \text{for} \ 0 \leq i \leq N-M-1.& 
\end{cases} 
\end{eqnarray}
In these new variables the monodoromy matrix \eqref{monodr1} can be written as 
\begin{equation} \label{monodrg}
\mathcal{M}(\lambda)=\prod_{i=0}^{M-1}\tilde{{\bf L}}(g_{r_i+N},\lambda-\alpha)\tilde{{\bf L}}(g_{r_i+N-M},\lambda)\tilde{{\bf L}}(g_{r_i+N-2M},\lambda)
\cdots \tilde{{\bf L}}(g_{r_{i+1}},\lambda) 
\end{equation}
and the $(N,M)$ KdV periodic reduction as  
\begin{equation*} 
\tilde{\phi}=f \circ \phi \circ f^{-1},
\end{equation*}
where $\phi$ is the map  \eqref{kdvphi} and  $f:\mathbb{C}^{N+M} \rightarrow \mathbb{C}^{N+M}$ is the 
birational change of coordinates  
$$f(v_0,v_1,\dots,v_{N+M-1}):=(g_0,g_1,\dots,g_{N+M-1}).$$
The explicit form of the map $\tilde{\phi}$ is  
\begin{eqnarray}
&& \tilde{\phi}(g_0,\dots,g_{N+M-1}) = \label{mapg} \\ 
&& \left(g_1,g_2,\dots,g_{N-M-1},g_{N-M}-\frac{\alpha}{g_N},g_{N-M+1},g_{N-M+2},\dots, g_{N+M-1},g_0+\frac{\alpha}{g_N}\right). \nonumber
\end{eqnarray}

We can also express the three invariant (compatible) Poisson brackets of the $(N,M)$ periodic reduction that were presented in the previous section in terms of the 
$g$-variables. 
The pushforward of the Poisson bracket \eqref{poi3} by the function $f$ implies the following Poisson bracket in $g$-variables 
(we use the same symbol $\{ \ , \ \}_3$  for this bracket as well)
\begin{eqnarray} \label{poissong3}
\{ g_i, g_j \}_3 = \begin{cases}
-\alpha, &j-i = M, \\ 
\alpha, &j-i = N, \\ 
0, & otherwise,
\end{cases} 
\end{eqnarray}
for $0\leq i<j \leq N+M-1$. 
This bracket is invariant under the map $\tilde{\phi}$.

\begin{remark}
The g-variables that we introduced here are not the same with the g-variables that appear in 
\cite{HKQT} 
for the $(N,1)$ periodic reductions of the lattice KdV equation.
\end{remark}

\subsection{Connection with the dressing chain and integrability}
We now consider, for any $K$ odd, the system of ordinary differential equations 
\begin{equation} \label{drsystem}
\dot{h}_i=-h_i(h_{i+1}-h_{i+2}+h_{i+3}-\dots -h_{i+K-1})+ b_i-b_{i+1},
\end{equation}
where the indices are considered modulo $K$, labelled 
$1,\ldots, K$. 
This system  was  
introduced by Veselov and Shabat in \cite{VesShab} from the  
dressing chain 
for Schr\"{o}dinger operators, and they  proved that this  is a Liouville integrable Hamiltonian system. 

\begin{proposition} 
The integrals of the dressing chain \eqref{drsystem} are given by the trace of the monodromy matrix   
\begin{equation} \label{dressmon}
\mathcal{K}(\lambda)=\tilde{{\bf L}}(h_K,\zeta_{K})\tilde{{\bf L}}(h_{K-1},\zeta_{K-1})\cdots \tilde{{\bf L}}(h_1,\zeta_{1}),
\end{equation}
where $\tilde{{\bf L}}$ is the matrix \eqref{laxg} and $\zeta_i=b_i-\lambda$.
\end{proposition} 

\begin{proof}
The trace of $\mathcal{K}(\lambda)$ is 
\begin{equation} \label{trdr}
\mathrm{tr}\mathcal{K}(\lambda)=\prod_{i=1}^{K}\Big(1+\zeta_{i+1} \frac{\partial^2}{\partial_{h_i} \partial_{h_{i+1}}}\Big)\prod_{i=1}^{K}h_i, 
\end{equation}
and this formula coincides with the trace formula of the (different) monodromy that is given by Veselov, Shabat in \cite{VesShab}. Equation \eqref{trdr} is obtained by a corresponding trace formula in Lemma 4.3 of \cite{honeward} for the product 
$\prod_{i=1}^{p}T_i$, where 
 $T_i=(P L(g_i,\lambda_i) P^{-1})^T$, with  
$P =\left(
\begin{array}{cc}
0 & 1 \\
 1 & 0
\end{array}
\right).$

\end{proof}

According to \cite{VesShab}, the integrals of the dressing chain are pairwise in involution with respect to the Poisson bracket 
\begin{eqnarray} \label{Poissondr}
\{ h_i, h_j \} = \begin{cases}
-1, &j-i = 1, \\ 
\ 1, &j-i = K-1, \\ 
0, & otherwise,
\end{cases} 
\end{eqnarray}
for $0\leq i<j \leq K$. 
This Poisson structure has rank $K-1$ and the function $h_1+h_2+ \dots + h_K$ (the coefficient of the highest degree term of 
the polynomial \eqref{trdr}) is a Casimir function.

Next, we set $K=N+M$ and we change the $h_i$ variables of the dressing chain to $g_j$ variables by setting  
\begin{equation} \label{htog}
(h_1,h_2,h_3\ldots, h_{N+M})=(\bf{g_0,g_1,g_2},\ldots, \bf{g_{r_{M-1}}}), 
\end{equation}
where
$${\bf{g_i}}=(g_{r_{M-i}},g_{r_{M-i}+M},g_{r_{M-i}+2M},g_{r_{M-i}+3M},...,g_{r_{M-i-1}+N}).$$

\begin{lemma} \label{IntdrtoKdV}
Under the change of variables \eqref{htog}, the integrals of the map \eqref{mapg} coincide with the integrals of the dressing chain \eqref{drsystem} for $K=N+M$ and a particular 
choice of the parameters $b_1,\dots,b_K$. 
\end{lemma}

\begin{proof}
We consider $b_i=-\alpha$, if $h_i=g_{r_{i}+N}$ according to \eqref{htog}, and the rest of the parameters $b_j$ 
being zero. 
Then, by comparing \eqref{monodrg} with \eqref{dressmon} under \eqref{htog} we conclude that 
$\mathcal{M}(\lambda)=\mathcal{K}(-\lambda)$. Therefore, the coefficients of the polynomials 
$\mathrm{tr}\mathcal{M}_g(\lambda)$ and $\mathrm{tr}\mathcal{K}(\lambda)$ coincide up to a sign. 
\end{proof}

Now, we can  prove the complete integrability of the map \eqref{mapg} and subsequently of the lattice KdV periodic reductions in the odd case. 
\begin{theorem}
For any coprime $N, M$, with $N> M>1$ and $N + M$ odd, the $(N,M)$ periodic reduction of lattice KdV given 
by the map $\phi$ in (\ref{kdvphi})  
is Liouville integrable.
\end{theorem}

\begin{proof}
It suffices to show that the map $\tilde{\phi}$ \eqref{mapg} is Liouville integrable. 
We have already proved that this map is Poisson 
with respect to the Poisson bracket \eqref{poissong3}. By considering the change of variables \eqref{htog} we 
observe that, in terms of the bracket (\ref{Poissondr}),  
\begin{equation*}
\{h_i,h_j\}_3=\alpha \{h_i,h_j\}.
\end{equation*}
That means that the map 
$\eta :\mathbb{C}^{N+M} \rightarrow \mathbb{C}^{N+M}$, with 
$$\eta(g_0,g_1,\dots,g_{N+M-1}):=(h_1,h_2,\ldots,h_{N+M})$$
is a Poisson map. 
Now, we consider the map \eqref{mapg} in the $h$-variables, given by conjugation  
$\eta \circ \tilde{\phi} \circ \eta^{-1}.$
 By construction, the latter  map is a Poisson map with respect to the bracket \eqref{Poissondr}. Furthermore, 
 by Lemma \ref{IntdrtoKdV} and by the Liouville integrability of the dressing chain we conclude that 
this map  is Liouville integrable and 
 consequently the map $\tilde{\phi}$ is Liouville integrable as well, as is $\phi$ which is conjugate to it by a 
birational transformation. 
\end{proof}

\begin{remark}
In a similar way we can express the first two Poisson brackets \eqref{poisbrack1} and \eqref{poisbrack2} in $g$-variables by \eqref{gvar} and consequently 
in  $h$-variables by \eqref{htog}. The induced compatible Poisson structures in $h$-variables agree with the bi-Hamiltonian formulation 
presented in \cite{VesShab} (see also \cite{Evri1}, where a relation of the dressing chain with integrable deformations of the Bogoyavlenskij-Itoh systems is established). Likewise, if we denote by $\Pi_2$ and $\Pi_3$ the Poisson bivectors of  \eqref{poisbrack2} and  \eqref{poi3} 
respectively, then it can be shown that 
\begin{equation*}
(\Pi_2+\frac{\lambda}{\alpha}\Pi_3)^{\#}(\mathrm{d}\,\mathrm{tr} \mathcal{M}(\lambda))=0.
\end{equation*}

Furthermore, we remark that the case of $(N,1)$ periodic reductions, for $N$ even, can be treated in the same way. In 
this case by just setting $h_{i}=g_{i-1}$, for $i=1,\ldots, N+1$, we obtain the equivalent dressing chain system with 
$b_1=b_2=\ldots b_n=0$ and $b_{N+1}=-\alpha$.

\end{remark}

Since the brackets $\{,\}_1$ and $\{,\}_2$ were obtained by lifting the log-canonical brackets of the form 
(\ref{ubrackets}) for the U-systems, the commuting first integrals for the map $\phi$ can be rewritten 
in terms of the variables $u_j$ and the parameters $\alpha$, $\beta$ so they provide commuting integrals 
for the U-systems, leding to the following.

\begin{corollary}
The maps $\hat{\varphi}_1$, $\hat{\varphi}_2$ in  $\mathbb{C}^{N+M-1}$ 
that correspond to the U-systems 
\eqref{1o} and \eqref{2o} are Liouville integrable.
\end{corollary}

\begin{remark} A careful counting shows that $\phi$ has one more integral than is necessary for 
the U-systems. It turns out that the Casimir of   $\{,\}_1$ or $\{,\}_2$ is the extra integral,  
and becomes a trivial function of the parameters when rewritten in the U-system variables. For 
explicit examples of this see \cite{HKW}.  
\end{remark}

\subsection{Poisson bracket on the space of monodromy matrices} 

A direct calculation shows that for $M+N$ odd the monodromy matrix can be written as 
\begin{equation}\label{mdy}
\mathcal{M}(\lambda)  = 
\left( \begin{array}{cc}
P(\lambda) & Q(\lambda) \\ 
R(\lambda) & S(\lambda) 
\end{array} 
\right) , \qquad Q(\lambda )=\lambda Q^*(\lambda), \quad S(\lambda )=\lambda S^*(\lambda),  
\end{equation} 
where $P$ has degree $\bar{g}=(M+N-1)/2$, 
$Q^*$ is monic of degree $\bar{g}$, 
$R$ is monic of degree $\bar{g}$, and 
$S^*$ has degree $\bar{g}-1$. 
For the second Poisson bracket $\{ , \}_2$, the entries of $\mathcal{M}$ satisfy a quadratic Poisson algebra, 
defined by 
$$ 
\{ R(\lambda), \nu^{-1}S(\nu) \}_2  =  \frac{R(\lambda) S(\nu) - R(\nu) S(\lambda)}{\lambda - \nu }, 
$$ 
$$
  \{ \lambda^{-1}Q(\lambda), \nu^{-1}P(\nu) \}_2  =  \frac{\lambda^{-1}Q(\lambda) P(\nu) - \nu^{-1}Q(\nu) P(\lambda)}{\lambda - \nu },
$$ 
$$  
\{ \lambda^{-1}Q(\lambda), \nu^{-1}S(\nu) \}_2  =  - \frac{\lambda^{-1}Q(\lambda) S(\nu) - \nu^{-1}Q(\nu) S(\lambda)}{\lambda - \nu },
$$
$$   
\{ R(\lambda), \nu^{-1}P(\nu) \}_2  =  - \frac{R(\lambda) P(\nu) - R(\nu) P(\lambda)}{\lambda - \nu }, 
$$ 
$$  
\{ \lambda^{-1}Q(\lambda), R(\nu) \}_2  =  \frac{S(\lambda) P(\nu) - S(\nu) P(\lambda)}{\lambda - \nu },
$$ 
$$  
\{ \lambda^{-1}P(\lambda), \nu^{-1}S(\nu) \}_2  =  -\frac{\lambda^{-1}Q(\lambda) R(\nu) - \nu^{-1}Q(\nu) R(\lambda)}{\lambda - \nu },
$$ 
together with 
$$ 
\{P(\lambda), P(\nu)\}_2 = \{Q(\lambda), Q(\nu)\}_2 =\{R(\lambda), R(\nu)\}_2 =\{S(\lambda), S(\nu)\}_2 =0, 
$$ 
so that the coefficients of the polynomial $P$ all Poisson commute with one another, and the same is true  
for the coefficients of the polynomials $Q,R,S$. 

If we write 
$$ 
Q^*=\prod_{i=1}^{\bar{g}} (\lambda - \lambda_i), 
\qquad 
=\prod_{i=1}^{\bar{g}} (\lambda - \zeta_i), 
$$ 
and set 
$$ \nu_i = S(\lambda_i). \qquad \eta_i = P(\zeta_i), 
$$ 
then from the spectral curve (\ref{spectral}) written as 
$$ 
\nu^2 - (P(\lambda)+S(\lambda)) \, \nu +P(\lambda) S(\lambda) - Q(\lambda) R(\lambda) = 0,  
$$
which is hyperelliptic and of genus $\bar{g}$, we see that 
$$ 
p_j = (\lambda_j,\nu_j), \quad \tilde{p}_j = (\zeta_j,\eta_j), \quad j=1,\ldots ,\bar{g} 
$$ 
are points on the curve. The poles and zeros of function 
$$ 
\frac{Q(\lambda)}{\nu -P(\lambda)} =  \frac{\nu -S(\lambda)} {R(\lambda)} 
$$ 
give the linear equivalence of divisors 
$$ 
(0,0)+ \sum_{j=1}^{\bar{g}} p_j \sim (\infty,\infty)+ \sum_{j=1}^{\bar{g}} \tilde{p}_j. 
$$ 
Moreover, the brackets for the monodromy matrix imply that the coordinates of these points (or rather, their 
logarithms) provide two sets of canonically conjugate variables on the phase space: 
$$ 
\{ \lambda_i,\nu_j\} =  \lambda_i \nu_i\delta_{ij}, \qquad 
\{ \zeta_i,\eta_j\} =  \zeta_i \eta_i\delta_{ij}, 
$$ 
with $\{ \lambda_i,\lambda_j\} = 0 =   \{\nu_i,\nu_j\}$, $\{ \zeta_i,\zeta_j\} = 0 =   \{\eta_i,\eta_j\}$. 

It is known that the algebro-geometric solutions of the discrete Hirota equation are given in terms 
of the Fay trisecant identity for an arbitrary algebraic curve  \cite{finitegenus}. 
It would be interesting to use the above spectral coordinates on the hyperelliptic curves (\ref{spectral}) to derive explicit formulae for the solutions of the 
iterated maps corresponding to 
the lattice KdV travelling wave reductions, as has been done for solutions of the discrete potential KdV equation in    
\cite{xunijhoff}.

\section{The even case}

We will now investigate the case when $N$, $M$ are odd and coprime. The situation in this case is different than before because the 
corresponding $U$-systems cannot provide an invariant Poisson structure for the KdV periodic reductions. However, they can do it for 
a third map that is one dimension higher than the associated $U$-systems and one dimension lower than the KdV map. The Liouville 
integrability of the latter map ensures the integrability of the $U$-systems and of the KdV periodic reductions.   

For $N,M$ odd and coprime we consider the Hirota reductions with periodic coefficients \eqref{1HKdV} and  \eqref{2HKdV}: 
\begin{eqnarray*}
\tau_{m+2N+M}\tau_m &=& \beta_m \tau_{m+N+M}  \tau_{m+N} - \alpha \tau_{m+2 N}\tau_{m+M},  \ \beta_{m+M}=\beta_m, \\
\tau_{m+2M+N}\tau_{m}&=& \beta'_m \tau_{m+N+M}  \tau_{m+M}+\alpha \tau_{m+2 M}\tau_{m+N},\  \beta'_{m+N}=\beta'_m.
\end{eqnarray*}    
The corresponding $U$-systems 
\begin{eqnarray} 
u_m u_{m+2} \ldots u_{m+N+M-2} &=& \beta_m- \alpha u_{m+M} u_{m+M+2} \ldots u_{m+N-2}, \ \beta_{m+M}=\beta_m,  \label{1Uev}\\ 
u'_m u'_{m+2} \ldots u'_{m+N+M-2} &=& \beta'_m+ \frac{\alpha}{ u'_{m+M} u'_{m+M+2} \ldots u'_{m+N-2}}, \ \  \beta'_{m+N}=\beta'_m\label{2Uev}
\end{eqnarray}   
are obtained by considering  
$$u_m=\frac{\tau_{m} \tau_{m+N+2}}{\tau_{m+2} \tau_{m+N}}, \ u'_m=\frac{\tau_{m} \tau_{m+M+2}}{\tau_{m+2} \tau_{m+M}},$$
respectively. 
On the other hand, the substitution $v_m=\frac{\tau_{m} \tau_{m+N+M}}{\tau_{m+M} \tau_{m+N}}$ leads to the KdV periodic reduction \eqref{RedkdV}, 
\begin{equation*} 
v_{m+{N}+{M}}-v_m=\alpha(\frac{1}{v_{m+N}}-\frac{1}{v_{m+M}}).
\end{equation*}
From the above substitutions it is not hard to derive that the $U$-systems variables and the KdV variables satisfy  
\begin{eqnarray*}
v_m v_{m+1} = u_m u_{m+1} \dots u_{m+M-1}=u'_m u'_{m+1} \dots u'_{m+N-1}.
\end{eqnarray*}
Next, we consider a new set of variables by setting $w_m=v_m v_{m+1}$ and we prove the following proposition. 

\begin{proposition}
Let $N,M$ be odd and co-prime. If $w_m,v_m,u_m,u'_m$ satisfy 
\begin{equation} \label{evvarall}
w_m=v_m v_{m+1} = u_m u_{m+1} \dots u_{m+M-1}=u'_m u'_{m+1} \dots u'_{m+N-1}, 
\end{equation}
then the following statements are equivalent: 
\begin{enumerate}
\item[(i)]  $v_m$ satisfies the $(N,M)$ KdV periodic reduction \eqref{RedkdV} 
\item[(ii)]  $u_m$ satisfies the $U$-system  \eqref{1Uev}
\item[(iii)]  $u'_m$ satisfies the $U$-system \eqref{2Uev}
\item[(iv)]  $w_m$ satisfies the recurrence  
\end{enumerate}
\begin{equation} \label{wsystem}
\prod_{i=0}^{\frac{N+M-2}{2}}w_{m+2i+1}-\prod_{i=0}^{\frac{N+M-2}{2}}w_{m+2i}=
\alpha(\prod_{i=0}^{\frac{N-3}{2}}w_{m+2i+1}\prod_{i=0}^{\frac{M-3}{2}}w_{m+2i+N+1}-
\prod_{i=0}^{\frac{M-3}{2}}w_{m+2i+1}\prod_{i=0}^{\frac{N-3}{2}}w_{m+2i+M+1})
\end{equation}

\end{proposition}

\begin{proof}
From \eqref{evvarall} we obtain 
\begin{equation} \label{extrarel1}
\frac{u_m}{u_{m+M}}=\frac{u'_m}{u'_{m+N}}=\frac{v_m}{v_{m+2}}, 
\end{equation}
and 
\begin{equation} \label{extrarel2}
\prod_{i=0}^{M-1}u_{m+2i}=v_m v_{m+M}, \ \prod_{i=0}^{N-1}u'_{m+2i}=v_m v_{m+N}. 
\end{equation}
The first $U$-system \eqref{1Uev} can be written as  
\begin{equation*}
(u_m u_{m+2} \dots u_{m+2M-2})\frac{u_{m+2M} u_{m+2M+2} \ldots u_{m+N+M-2}}{ u_{m+M} u_{m+M+2} \ldots u_{m+N-2}}+
\alpha =\frac{\beta_m}{ u_{m+M} u_{m+M+2} \ldots u_{m+N-2}}
\end{equation*}
and from (\ref{extrarel1}-\ref{extrarel2}) as 
$$ v_m v_{m+M} \frac{v_{m+N}}{v_{m+M}}+ \alpha= \frac{\beta_m}{u_{m+M} u_{m+M+2} \ldots u_{m+N-2}}.$$
By the periodicity condition $\beta_m=\beta_{m+M}$ we get 
\begin{eqnarray*}
v_m {v_{m+N}}+ \alpha &=& \frac{\beta_{m+M}}{u_{m+M} u_{m+M+2} \ldots u_{m+N-2}}\\ 
 &=&u_{m+N} u_{m+N+2} \dots u_{m+N+2M-2}+ \alpha \frac{u_{m+2M} u_{m+2M+2} \ldots u_{m+M+N-2}}{{u_{m+M} u_{m+M+2} \ldots u_{m+N-2}}} \\ 
 &=& v_{m+N}v_{m+N+M}+ \alpha \frac{v_{m+N}}{v_{m+M}}  
 \end{eqnarray*}
(by using (\ref{extrarel1}-\ref{extrarel2}) in the last equality), which implies the KdV reduction \eqref{RedkdV}. Conversely, 
 \eqref{RedkdV} implies $\beta_m=\beta_{m+M}$, for 
$ \beta_m=u_m u_{m+2} \ldots u_{m+N+M-2}+\alpha u_{m+M} u_{m+M+2} \ldots u_{m+N-2}$. 

The second  $U$-system \eqref{2Uev} can be written as
\begin{equation*}
u'_m u'_{m+2} \ldots u'_{m+2N-2}-\frac{\alpha u'_{m+N+M} u'_{m+N+M+2} \ldots u'_{m+2N-2}}{ u'_{m+M} u'_{m+M+2} \ldots u'_{m+N-2}} = 
\beta'_m u'_{m+N+M} u'_{m+N+M+2} \ldots u'_{m+2N-2}.
\end{equation*}
By (\ref{extrarel1}-\ref{extrarel2}) and the  periodicity condition $\beta'_m=\beta'_{m+N}$ we derive 
\begin{eqnarray*}
v_m v_{m+N} -\alpha \frac{v_{m+N}}{v_{m+M}}&=&\beta'_m u'_{m+N+M} u'_{m+N+M+2} \ldots u'_{m+2N-2} \\&=&
\beta'_{m+N} u'_{m+N+M} u'_{m+N+M+2} \ldots u'_{m+2N-2}  \\
&=& ( u'_{m+M} u'_{m+M+2} \ldots u'_{m+N-2})( u'_{m+N} u'_{m+N+2} \ldots u'_{m+2N+M-2})\\ && \cdot 
\frac{u'_{m+N+M} u'_{m+N+M+2} \ldots u'_{m+2N-2}}{u'_{m+M} u'_{m+M+2} \ldots u'_{m+N-2}}-\alpha
=v_{m+N+M}v_{m+N}-\alpha
\end{eqnarray*}
and conversely from \eqref{RedkdV} we conclude that $\beta'_m=\beta'_{m+N}$ for 
$$\beta'_m=u'_m u'_{m+2} \ldots u'_{m+N+M-2}- \frac{\alpha}{ u'_{m+M} u'_{m+M+2} \ldots u'_{m+N-2}}.$$

Finally, \eqref{RedkdV} is equivalent to 
\begin{equation*}
\prod_{i=0}^{N+M-1}v_{m+i+1}-\prod_{i=0}^{N+M-1}v_{m+i}=
\alpha(\frac{\prod_{i=1}^{N+M-1}{v_{m+i}}}{v_{m+N}}
-\frac{\prod_{i=1}^{N+M-1}{v_{m+i+1}}}{v_{m+M}}), 
\end{equation*}
which up to  \eqref{evvarall} is equivalent to \eqref{wsystem}.
\end{proof}

We claim that in this case the corresponding map of the recurrence  \eqref{wsystem} inherits bi-Poisson structure and it is Liouville integrable. 
We will demonstrate this in the following example for $N=5$ and $M=3$.

\begin{remark}
In the case of $M=1$, \eqref{evvarall} implies that $\omega_m=u_m$, i.e. recurrence \eqref{wsystem} coincides with the U-system \eqref{1Uev} that corresponds to the Hirota reduction 
\eqref{1HKdV}. 
Therefore,   U-system \eqref{1Uev} inherits a bi-Poisson structure which proves the integrability of these cases.

The integrability of KdV reductions with $M=1$ 
has been proved in \cite{HKQT}, using the observation that these reductions are given in terms of a tau function that satisfies the bilinear recurrence relation of the 
form \eqref{2HKdV} (for $M=1$) and a Poisson structure derived by the Lagrangian formulation of the reduced maps.  The results in this work extend this observation to the
general $(N,M)$ KdV reductions, and shows that in each case there are actually two different bilinear
equations involved  (\eqref{1HKdV} and  \eqref{2HKdV}) which provide the two compatible Poisson structures.
\end{remark}

\subsection{A Periodic KdV reduction of order 8}  
We consider the case of $N=5$, $M=3$. The two bilinear equations in this case are  
\begin{eqnarray}
\tau_{m+13}\tau_m &=& \beta_m \tau_{m+8}  \tau_{m+5} - \alpha \tau_{m+10}\tau_{m+3},  \  \beta_{m+3}=\beta_m, \\
\tau_{m+11}\tau_{m}&=& \beta'_m \tau_{m+8}  \tau_{m+3}+\alpha \tau_{m+6}\tau_{m+5}, \ \  \beta'_{m+5}=\beta'_m.
\end{eqnarray}
For 
$u_m=\frac{\tau_{m} \tau_{m+7}}{\tau_{m+2} \tau_{m+5}}, \ u'_m=\frac{\tau_{m} \tau_{m+5}}{\tau_{m+2} \tau_{m+3}},$ 
we obtain the corresponding $U$-systems
\begin{eqnarray}
u_m u_{m+2} u_{m+4} u_{m+6} &=& \beta_m- \alpha u_{m+3}, \ \  \beta_{m+3}=\beta_m,  \label{1Uex} \\ 
u'_m u'_{m+2} u'_{m+4} u'_{m+6} &=& \beta'_m+ \frac{\alpha}{ u'_{m+3}}, \ \  \ \beta'_{m+N}=\beta'_m\label{2Uex}
\end{eqnarray} 
and for 
 $v_m=\frac{\tau_{m} \tau_{m+8}}{\tau_{m+5} \tau_{m+3}}$ the $(5,3)$ reduction of the lattice KdV equation 
\begin{equation} \label{KdV53}
v_{m+8}-v_m=\alpha(\frac{1}{v_{m+5}}-\frac{1}{v_{m+3}}).
\end{equation}
Furthermore, by setting 
\begin{equation} \label{wvar}
w_m=v_m v_{m+1} = u_m u_{m+1}  u_{m+2}=u'_m u'_{m+1}u'_{m+2} u'_{m+3} u'_{m+4}, 
\end{equation}
we derive the recurrence 
\begin{equation} \label{wsystex}
w_{m+4} w_{m+6}(w_m w_{m+2}-\alpha w_{m+1})=w_{m+1} w_{m+3}(w_{m+5} w_{m+7}-\alpha w_{m+6}).
\end{equation}
We will denote by $\phi^{(i)}_m$, $i=1,2$ the associated maps with the $U$-systems and by $\phi_{v}$, $\phi_w$ the 
maps that correspond to the recurrences \eqref{KdV53} and \eqref{wsystem} respectively, that is     
\begin{eqnarray*}
&& \phi^{(1)}_m(u_0,u_1,u_2,u_3,u_4, u_5) = (u_1,u_2, u_3, u_4,u_5,\frac{ \beta_m- \alpha u_{3}}{u_0 u_{2} u_{4}} ), \\ 
&& \phi^{(2)}_m(u'_0,u'_1,u'_2,u'_3,u'_4, u'_5) = (u'_1,u'_2, u'_3, u'_4,u'_5,\frac{\alpha + \beta'_m u'_3  }{u'_0 u'_{2} u'_3 u'_{4}} ), \\ 
&& \phi_v(v_0,v_1,v_2,v_3,v_4,v_5, v_6, v_7)  =   (v_1,v_2,v_3,v_4,v_5, v_6, v_7,v_0+\alpha(\frac{1}{v_{5}}-\frac{1}{v_{3}})), \\
&& \phi_w(w_0,w_1,w_2,w_3,w_4,w_5, w_6) = (w_1,w_2, w_3,w_4,w_5, w_6, 
w_6 \frac{\alpha w_1(w_3-w_4)+w_0 w_2 w_4}{w_1 w_3 w_5}).
\end{eqnarray*}

\subsubsection{Compatible Poisson structures} 

The maps $\phi^{(1)}_m$, $ \phi^{(2)}_m$  associated with the $U$-systems are symplectic with respect to the symplectic structure 
specified by the Poisson brackets 
\begin{align}
  \{u_i,u_{i+3}\}_1 &=u_i u_{i+3},  \ \{u_i,u_{i+5}\}_1=-u_i u_{i+5}, \ \{u_i,u_{i+j}\}_1 =0, \ \text{for} \ j=1,2,4, \    \label{u1pbr}\\
 \{u'_i,u'_{i+3}\}_2 &=u'_i u'_{i+3},  \ \{u'_i,u'_{i+5}\}_2=-u'_i u'_{i+5}, \ \{u'_i,u'_{i+j}\}_1 =0, \ \text{for} \ j=1,2,4. \label{u2pbr}
\end{align}
We will show that these Poisson structures for the $U$-systems give rise to two different Poisson structures for \eqref{wsystex}. 

First, by considering $w_m=u_m u_{m+1} u_{m+2}$ and the recurrence  \eqref{1Uex}, we can write the $w$ variables of the map $\phi_w$  as  
\begin{eqnarray} \label{subw1}
&& w_0=u_0 u_1 u_2, \ w_1=u_1 u_2 u_3, \ w_2=u_2 u_3 u_4, \ w_3=u_3 u_4 u_5, \\
&& w_4=\frac{(\beta_0-\alpha u_3)u_5}{ u_0 u_2}, \ 
 w_5=\frac{(\beta_0-\alpha u_3)(\beta_1-\alpha u_4)}{ u_0 u_1 u_2 u_3 u_4}, \ w_6=\frac{(\beta_1-\alpha u_4)(\beta_2-\alpha u_5)}{ u_1 u_2 u_3 u_4 u_5}. \nonumber
\end{eqnarray}
Now, we can evaluate the Poisson brackets $\{w_i.w_j \}_1$ from \eqref{u1pbr} and  \eqref{subw1} to derive

\begin{eqnarray} \label{PBw1}
&& \{w_i,w_{i+1}\}_1=w_i w_{i+1}, \ \{w_i,w_{i+2}\}_1=2w_i w_{i+2}, \ \{w_i,w_{i+3}\}_1=2w_i w_{i+3}, \\
 && \{w_i,w_{i+4}\}_1=-\frac{\alpha w_{i+1}  w_{i+3}}{w_{i+2}},   \   
   \{w_i,w_{i+5}\}_1=-2w_i w_{i+5}-\alpha(\frac{ w_i  w_{i+2}w_{i+4}}{w_{i+1}w_{i+3}}+\frac{ w_{i+1}  w_{i+3}w_{i+5}}{w_{i+2}w_{i+4}}), \nonumber \\ 
 &&   \{w_i,w_{i+6}\}_1=-2w_i w_{i+6}-\alpha \frac{ w_i  w_{i+2}w_{i+4} w_{i+6}}{w_{i+1}w_{i+3}w_{i+5}}. \nonumber
\end{eqnarray}

Furthermore, by setting $w_m=u'_m u'_{m+1}u'_{m+2} u'_{m+3} u'_{m+4}$, from \eqref{2Uex}, we obtain 
\begin{eqnarray} \label{subw2}
&& w_0=u'_0 u'_1 u'_2 u'_3 u'_4, \ w_1=u'_1 u'_2 u'_3 u'_4 u'_5,  \\
&& w_2=(\alpha+\beta_0' u'_3) \frac{w_1}{w_0}, \ w_3=(\alpha+\beta_1' u'_4) \frac{w_2}{w_1}, \ w_4=(\alpha+\beta'_2 u'_5) \frac{w_3}{w_2},\nonumber \\ 
&& w_5=(\alpha+\frac{\beta'_3(\alpha+\beta'_0 u'_3)}{u'_0 u'_2 u'_3u'_4} )\frac{w_4}{w_3},  \  
w_6=(\alpha+\frac{\beta'_4(\alpha+\beta'_1 u'_4)}{u'_1 u'_3 u'_4u'_5} )\frac{w_5}{w_4}.\nonumber 
\end{eqnarray}
Similarly, from \eqref{u2pbr} and \eqref{subw2} we evaluate  the brackets $\{w_i.w_j \}_2$ in terms of the $w$ variables that gives 
\begin{eqnarray} \label{PBw2}
&& \{w_i,w_{i+1}\}_2=w_i w_{i+1}, \ \{w_i,w_{i+2}\}_2=-\alpha w_{i+1}+2w_i w_{i+2}, \\
 && \{w_i,w_{i+3}\}_2=2w_i w_{i+3}-\alpha(\frac{ w_i  w_{i+2}}{w_{i+1}}+\frac{ w_{i+1}  w_{i+3}}{w_{i+2}}),   \nonumber \\
&&  \{w_i,w_{i+4}\}_2=-\frac{\alpha w_i  w_{i+2}w_{i+4}}{w_{i+1}w_{i+3}},  \ 
 \{w_i,w_{i+5}\}_1=-2w_i w_{i+5}+\frac{\alpha^2  w_{i+1}w_{i+4}}{w_{i+2}w_{i+3}},  \nonumber \\ 
 && \{w_i,w_{i+6}\}_2=-2w_i w_{i+6}+\alpha^2( \frac{ w_i  w_{i+2}w_{i+5}}{w_{i+1}w_{i+3}w_{i+4}}+\frac{ w_{i+1}  w_{i+4}w_{i+6}}{w_{i+2}w_{i+3}w_{i+5}}). \nonumber
\end{eqnarray}

\begin{proposition}
The brackets \eqref{PBw1},  \eqref{PBw2} define two compatible Poisson brackets on $\mathbb{C}^7$ (with coordinates $\{w_0,\ldots, w_6 \}$) and 
the birational  map $\phi_w:\mathbb{C}^7 \rightarrow \mathbb{C}^7$ that corresponds to the recurrence \eqref{wsystex} preserves both of them. 
\end{proposition}

\begin{proof}
From the construction of  \eqref{PBw1} and  \eqref{PBw2}, it follows directly that they are Poisson brackets. Furthermore, the bracket 
$\{ \ , \ \}^w_3=\{ \ , \ \}_1-\{ \ , \ \}_2$ satisfies the Jacobi identity, so  \eqref{PBw1} and  \eqref{PBw2} are compatible. Finally, the preservation of these brackets under the map $\phi_w$ follows from the preservation of the Poisson brackets \eqref{u1pbr} and \eqref{u2pbr} under the maps 
$\phi^{(1)}_m$ and $\phi^{(2)}_m$ respectively. 
\end{proof}

\subsubsection{Monodromy matrix and integrability} 

The monodromy matrix in terms of the KdV coordinates is derived by \eqref{monodr1} for $N=5$, $M=3$ and it reads 
\begin{equation} \label{monex}
\mathcal{M}={\bf M}(v_5){\bf L}(v_2,v_5){\bf M}(v_7){\bf L}(v_4,v_7){\bf L}(v_1,v_4){\bf M}(v_6){\bf L}(v_3,v_6){\bf L}(v_0,v_3), 
\end{equation}
where ${\bf L}({v}_i,{v}_j)$ and ${\bf M}({v_i})$ are given by \eqref{LaxPair2}. 
The trace of the Monodromy matrix can be written as 
$$tr\mathcal{M}_v=2 \lambda^4+I_3 \lambda^3+I_2 \lambda^2+I_1 \lambda+I_0,$$
where $I_0,I_1,I_2,I_3$ are functionally independent first integrals of the map  $\phi_v$. From these integrals, by 
considering the substitution \eqref{wvar}, we obtain four integrals for the map $\phi_w$. Let us denote them by 
 $\tilde{I}_0,\tilde{I}_1,\tilde{I}_2,\tilde{I}_3$ respectively.

\begin{proposition}
The map $\phi_w: \mathbb{C}^7  \rightarrow \mathbb{C}^7$,  
$$ \phi_w(w_0,w_1,\ldots,w_5, w_6) = (w_1,w_2, \ldots, w_6, 
w_6 \frac{\alpha w_1(w_3-w_4)+w_0 w_2 w_4}{w_1 w_3 w_5})$$
is Liouville integrable. 
\end{proposition}

\begin{proof}
The rank of the Poisson brackets \eqref{PBw1} and 
 \eqref{PBw2} is six. Furthermore, the integrals  $\tilde{I}_0,\tilde{I}_1,\tilde{I}_2,\tilde{I}_3$ are functionally independent and pairwise in 
 involution with respect to both Poisson brackets. 
\end{proof}

The function 
\begin{eqnarray*}
\mathcal{C}_1 &=& \tilde{I}_0+\alpha \tilde{I}_1+\alpha^2 \tilde{I}_2+\alpha^3\tilde{I}_3+2 \alpha^4 \\ 
&=&\frac{(\alpha w_1 w_3+w_0 w_2 w_4)(\alpha w_2 w_4+w_1 w_3 w_5)(\alpha w_3 w_5+w_2 w_4 w_6)}{w_1 w_2 w_3 w_4 w_5}
\end{eqnarray*}
is a Casimir function for the Poisson bracket \eqref{PBw1}, while the function 
\begin{eqnarray*}
\mathcal{C}_2 \ = \ - \tilde{I}_0 \ = \ \frac{1}{w_3} \prod_{i=0}^{4} (\alpha-\frac{w_i w_{i+2}}{w_{i+1}})
\end{eqnarray*}
is a Casimir function for the Poisson bracket \eqref{PBw2}. 

The integrability of the map $\phi_v: \mathbb{C}^8 \mapsto \mathbb{C}^8$, that corresponds to the $(5,3)$ KdV periodic reduction, 
follows from the integrability of the map $\phi_w$. In particular, let us consider the Poisson bracket 
\begin{eqnarray*}
\{ v_i, v_j \}^v_3 &=& \begin{cases}
-\alpha, &j-i = 3, \\ 
 \ \alpha, &j-i = 5, \\ 
  \frac{\alpha^2} {v_{3+i}^{2}}, &j-i = 6, \\ 
 \ 0, & otherwise.
\end{cases} 
\end{eqnarray*}
for $0\leq i<j \leq 7$. This Poisson bracket coincides with the bracket \eqref{poi3} for $N=5$ and $M=3$. Now, we can see directly that 
the pushforward of the Poisson bracket $\{ \ , \ \}^v_3$ by the map 
$\pi_v:\mathbb{C}^8 \rightarrow \mathbb{C}^7$, 
$$  \pi_v(v_0, \ldots v_7)=(w_0, \ldots, w_6),  \ w_j=v_j v_{j+1},$$
yields the Poisson bracket $\{ \ , \ \}^w_3=\{ \ , \ \}_1-\{ \ , \ \}_2$, where $\{ \ , \ \}_1, \ \{ \ , \ \}_2$ are the brackets \eqref{PBw1} and 
\eqref{PBw2} respectively. In other words, the map $\pi_v:( \mathbb{C}^8, \{ \ , \ \}^v_3) \rightarrow (\mathbb{C}^7, \{ \ , \ \}^w_3)$ is 
a Poisson map. Therefore, the involution of the integrals $\tilde{I}_i$, $i=0,\ldots,3,$ implies the involution of the KdV integrals  
$I_i=\tilde{I}_i \circ \pi_v$. Furthermore, $\phi_v$ is a Poisson map with respect to $\{ \ , \ \}^v_3$. Hence, we conclude that $\phi_v$ 
is Liouville integrable as well. With similar arguments we can prove the Liouville integrability of the U-systems \eqref{1Uev} and 
\eqref{2Uev}. 

\section{Conclusions}
Plane wave type reductions of the discrete Hirota equation are associated with periodic reductions of integrable lattice equation. 
In this work, we focussed on a particular class of Hirota reductions associated with the $(N,M)$ periodic reductions of the lattice KdV equation and we studied the integrability of the induced ordinary difference equations using the properties of the underlying cluster algebra structure. We developed various integrability aspects of the corresponding maps, including invariant Poisson structures, a  bi-hamiltonian formalism, a refactorization of the monodromy matrices and a connection with the integrals of the dressing chain.  
In this way, we managed to prove the Liouville integrability of  all  lattice KdV periodic reductions and of the corresponding $U$-systems when $N+M$ is odd. The even case turned out to be more complicated. In this case the integrability can be justified by the integrability of a different map (one dimension lower than the KdV map) that inherits two compatible Poisson structures. We demonstrated this through an example for $N=5$ and $M=3$ and we aim to give a full proof in the future. In a similar framework we 
are motivated to study different families of discrete Hirota reductions associated with various integrable lattice equations.

\newpage

\appendix
\section{Poisson brackets for U-systems} \label{appa} 

In this appendix, we give a proof of Theorem \ref{dimthm}, and further give 
a precise description of the nondegenerate Poisson brackets for the U-systems in 
Proposition  \ref{redsymp}.                                                                                                                                                                                                                                                                                                                                                                                                                                                                                                                                                                                                                                                                                                                                                                                                                                                                                                                                                                                                                                                                                                                                                                                                                                                                                                                                                                                                                                                                                                                                                                                                                                                                                                                                                                                                                                                                                                                                                                                                                                                                                                                                                                                                                                                                                                                                                                                                                                                                                                                                                                                                                                                                                                                                                                                                                                                                                                                                                                                                                                                                                                                                                                                                                                                                                                                                                                                                                                                                                                                                                                                                                                                                                                                                                                                                                                                                                                                                                                                                                                                                                                                                                                                                                                 
The fact that these U-systems are symplectic when $N$ and $M$ are coprime follows from a computation of the rank  
of the exchange matrix $B$ associated with one of the T-systems (\ref{t1}) and (\ref{t2}).

In order to illustrate the proof, we consider the $B$ matrix for the T-system  for (\ref{t1}) 
in the case $N=4$, $M=3$, which is the $11\times 11$ matrix 
\beq\label{bex2}
B=\left(
\begin{array}{ccccccccccc}
0 & 0 & 0 & 1 & -1 & 0  & 0 &  -1 & 1 & 0 & 0\\
0 & 0 & 0 & 0 &  1  & -1 &  0 & 0 &-1 & 1 & 0\\
0 & 0 & 0 & 0 &  0 &   1 & -1 &  0 & 0 &-1& 1\\
-1&0&0&0&1&0&1 & 0 & 0 &0&-1\\
1&-1&0&-1&0&1&0&1&-1&0 &0\\
0& 1&-1&0&-1&0&1&0&1&-1&0 \\
0& 0& 1&-1&0&-1&0&1&0&1&-1 \\
1&0&0&0&-1&0&-1 & 0 & 0 &0&1\\
-1&1&0&0& 1&-1&0&0&0&0&0\\
0&-1&1&0&0&1&-1&0&0&0&0\\
0&0&-1&1&0& 0 & 1&-1&0&0&0
\end{array}
\right).
\eeq
In this case, $B$ has rank 6, and 
according to Proposition 3.9 in \cite{honeinoue}, 
im$\,B$ has a palindromic basis of integer vectors, unique up to fixing an overall   
sign. This basis is obtained by starting from the vector  
$$
{\bf w}_1=(1,-1,0,0, -1,1,0,0,0,0,0)^T\in \mathbb{Z}^{11},$$ 
and then repeatedly applying the shift operator s, which 
shifts the non-zero entries of any vector  so that 
$$
{\bf w}_2=\rs ({\bf w}_1)=(0, 1,-1,0,0, -1,1,0,0,0,0)^T,$$ 
and so on, up to   
$${\bf w}_6=\rs^5 ({\bf w}_1)=(0,0,0,0, 0,1,-1,0,0, -1,1)^T.$$
Letting ${\bf r}_j$ denote the $j$th row of $B$, we see that   
the first and last three rows are given by 
$$
{\bf r}_1^T ={\bf w}_4, \, 
{\bf r}_2^T ={\bf w}_5, \,
{\bf r}_3^T ={\bf w}_6, \,
{\bf r}_9^T =-{\bf w}_1, \,
{\bf r}_{10}^T =-{\bf w}_2, \,
{\bf r}_{11}^T =-{\bf w}_3,  $$
while the middle three rows are written in terms of the basis as 
$$ 
{\bf r}_5^T ={\bf w}_1-{\bf w}_4, \quad 
{\bf r}_6^T ={\bf w}_2-{\bf w}_5, \quad
{\bf r}_7^T ={\bf w}_3-{\bf w}_6,
$$ 
and the the remaining rows on either side of these are given by 
$$ 
{\bf r}_4^T =-({\bf w}_1+{\bf w}_2+{\bf w}_3+{\bf w}_4+{\bf w}_5+{\bf w}_6), \qquad 
{\bf r}_8^T ={\bf w}_1+{\bf w}_2+{\bf w}_3+{\bf w}_4+{\bf w}_5+{\bf w}_6. 
$$ 
By direct calculations using the conditions (\ref{cond}) on 
the entries $(B_{ij})$, these formulae generalize to the exchange matrices of all the T-systems (\ref{t1}), 
of size $2N+M$, which are obtained by starting from 
the first row with $1$ in entries $M+1$ and $2N+1$, $-1$ in entries $N+1$ and 
$M+N+1$, and all other entries zero, and then recursively obtaining the 
subsequent rows from the conditions (\ref{cond}). Similar calculations apply 
to the exchange matrices of all the T-systems (\ref{t2}), 
of size $2M+N$.

\begin{lemma}\label{blem}
For $N+M$ odd, 
the rows of the exchange matrix $B$ of size $2N+M$ for the T-system (\ref{t1}) 
are given by 
$$ 
{\bf r}_k^T = \rs^{M+k-1} ({\bf v}), 
\quad 
 {\bf r}_{2N+k}^T = -\rs^{k-1} ({\bf v}), 
 \quad 
 {\bf r}_{N+k}^T = \rs^{k-1} ({\bf v}'), \quad 
 \mathrm{for}\quad k=1,\ldots,M, 
$$ 
$$ 
 {\bf r}_{M+k}^T = -\rs^{k-1} ({\bf v}'') = 
  - {\bf r}_{N+M+k}^T , \quad 
 \mathrm{for}\quad k=1,\ldots,N-M, 
$$ 
with 
$$ 
{\bf v} = \sum_{j=1}^{N-M} {\bf w}_j, 
\quad 
{\bf v}' = \sum_{j=1}^{M} {\bf w}_j-{\bf w}_{N-M+j}, 
\quad {\bf v}'' = \sum_{j=1}^{2M} {\bf w}_j, 
$$ 
where   the vectors ${\bf w}_j = \rs^{j-1} ({\bf w}_1)$, $j=1,\ldots, N+M-1$ are the shifts of the 
vector ${\bf w}_1$ with 1 in entries 1 and and $N+2$,  $-1$ in entries 2 and $N+1$, and all 
other entries zero.  
Similarly, the rows of the exchange matrix $B$ of size $2M+N$ for the T-system (\ref{t2}) 
are given by 
$$ 
{\bf r}_k^T = \rs^{M+k-1} ({\bf v}), 
\quad 
 {\bf r}_{N+M+k}^T = -\rs^{k-1} ({\bf v}), 
 \quad \mathrm{for}\quad k=1,\ldots,M, 
$$
$$ 
 {\bf r}_{\min (2M,N)+k}^T = \rs^{k-1} ({\bf v}'), \quad 
 \mathrm{for}\quad k=1,\ldots,\vert N-2M\vert, 
$$ 
$$ 
 {\bf r}_{M+k}^T = -\rs^{k-1} ({\bf v}'') = 
  - {\bf r}_{\max (2M,N) +k}^T,
 \quad 
 \mathrm{for}\quad k=1,\ldots,\min (M,N-M), 
$$ 
with 
$$ 
{\bf v} = \sum_{j=1}^{N-M} {\bf w}_j, 
\quad 
{\bf v}' = \sum_{j=1}^{\min (M,N-M)} {\bf w}_j-{\bf w}_{\min (2M,N) +j}, 
\quad {\bf v}'' = \sum_{j=1}^{2M} {\bf w}_j +\sum_{j=1}^{N-M} {\bf w}_{M+j}, 
$$ 
where in the latter case  the vectors ${\bf w}_j$ for  $j=1,\ldots, N+M-1$ are the shifts of the 
vector ${\bf w}_1$ with 1 in entries 1 and and $M+2$,  $-1$ in entries 2 and $M+1$, and all 
other entries zero. 
\end{lemma} 

The above explicit expressions for the rows of $B$ show that im$\,B$ is 
a 
subspace of the span of the vectors ${\bf w}_1, \ldots, {\bf w}_{N+M-1}$, but to show 
that these spaces coincide requires coprimality of $N$ and $M$. 
 
\begin{proposition} \label{brank} 
For $\gcd (N,M)=1$ with $N+M$ odd, the exchange matrices for the T-systems 
(\ref{t1}) and (\ref{t2}) both have rank $N+M-1$. 
\end{proposition} 
\begin{proof} 
To prove the result, it suffices to show that $N+M-1$ rows (or columns) of $B$ are linearly 
independent. For (\ref{t1}), where the exchange matrix has size $2N+M$, 
if we choose the first $M$ rows of $B$, and minus the last $M$ rows, together 
with rows ${\bf r}_{N+M+1},\ldots,{\bf r}_{2N-1}$, then by expanding them in the basis ${\bf w}_j$ 
we see that this gives $N+M-1$ independent vectors if and only if the determinant 
\begin{equation}
\label{delta}  
\begin{array}{c@{\!\!\!}l}
\left| 
\begin{array}{cccccccc} 
1 & \cdots & \cdots & 1 & 0 & \cdots & \cdots & \cdots \\ 
0 & 1 & \cdots  & \cdots & 1 & 0 & \cdots & \cdots \\ 
 & \ddots & \ddots  &   &  & \ddots &  &  \\ 
 & & \ddots & \ddots & & & \ddots & \\ 
  &   &         & 0 &  1 & \cdots & \cdots  & 1 \\ 
1 & \cdots & \cdots & \cdots & 1 & 0 & \cdots & \cdots \\ 
 & \ddots & & & & \ddots & \ddots & \\
 & & 1 & \cdots & \cdots & \cdots & 1 & 0 
\end{array}
\right|  
& 
\begin{array}[c]{@{}l@{\,}l}
   \left. \begin{array}{c} \vphantom{0}  \\ \vphantom{0}
   \\ \vphantom{\ddots} \\ \vphantom{\ddots} \\ \vphantom{0} \end{array} \right\} & \text{$2M$} \\
\left. \begin{array}{c} \vphantom{0} \\ \vphantom{\ddots}
   \\ \vphantom{0}  \end{array} \right\} & \text{$N-M-1$}
\end{array}
\end{array}
\end{equation}
is non-zero. For (\ref{t2}), where the exchange matrix has size $2M+N$, we 
again choose the first and the last $M$ rows, and apply suitable row 
operations (which are different for $N>2M$ and $N<2M$) to obtain $N-M-1$ more rows, 
leading  to the same determinant.  
Up to an overall sign, expanding (\ref{delta}) about the last column yields the determinant 
of the $(N+M-2)\times (N+M-2)$ Sylvester matrix for the resultant of two polynomials, namely 
$$ 
\mathrm{Res}\left(\frac{x^{2M}-1}{x-1}, \frac{x^{N-M}-1}{x-1} \right), 
$$ 
and this is non-zero if and only if $N$ and $M$ are coprime, when these  polynomials have no 
roots in common.  
\end{proof} 

Before considering the associated U-systems, we present the corresponding results when $N+M$ is even. 

\begin{lemma}\label{blemeven}
For $N+M$ even, 
the rows of the exchange matrix $B$ of size $2N+M$ for the T-system (\ref{t1}) 
are given by 
$$ 
{\bf r}_k^T = \rs^{M+k-1} ({\bf v}), 
\quad 
 {\bf r}_{2N+k}^T = -\rs^{k-1} ({\bf v}), \quad 
 {\bf r}_{N+k}^T = \rs^{k-1} ({\bf v}'), \quad 
 \mathrm{for}\quad k=1,\ldots,M, 
$$ 
$$ 
 {\bf r}_{M+k}^T = -\rs^{k-1} ({\bf v}'') = 
  - {\bf r}_{N+M+k}^T , \quad 
 \mathrm{for}\quad k=1,\ldots,N-M, 
$$ 
with 
$$ 
{\bf v} = \sum_{j=1}^{(N-M)/2} {\bf w}_{2j-1}, 
\quad 
{\bf v}' = \sum_{j=1}^{(N-M)/2} {\bf w}_{2j-1}-{\bf w}_{M+2j-1}, 
\quad {\bf v}'' = \sum_{j=1}^{M} {\bf w}_{2j-1}, 
$$ 
where   the vectors ${\bf w}_j = \rs^{j-1} ({\bf w}_1)$, $j=1,\ldots, N+M-2$ are the shifts of the 
vector ${\bf w}_1$ with 1 in entries 1 and $N+3$,  $-1$ in entries 3 and $N+1$, and all 
other entries zero.  
Similarly, the rows of the exchange matrix $B$ of size $2M+N$ for the T-system (\ref{t2}) 
are given by 
$$ 
{\bf r}_k^T = \rs^{M+k-1} ({\bf v}), 
\quad 
 {\bf r}_{N+M+k}^T = -\rs^{k-1} ({\bf v}), \quad 
 \mathrm{for}\quad k=1,\ldots,M, 
$$
$$ 
 {\bf r}_{\min (2M,N)+k}^T = \rs^{k-1} ({\bf v}'), \quad 
 \mathrm{for}\quad k=1,\ldots,\vert N-2M\vert, 
$$ 
$$ 
 {\bf r}_{M+k}^T = -\rs^{k-1} ({\bf v}'') = 
  - {\bf r}_{\max (2M,N) +k}^T 
\quad 
 \mathrm{for}\quad k=1,\ldots,\min (M,N-M), 
$$ 
with 
$$ 
{\bf v} = \sum_{j=1}^{(N-M)/2} {\bf w}_{2j-1},  \quad 
{\bf v}' = \begin{cases} \sum_{j=1}^{M} {\bf w}_{2j-1}-{\bf w}_{M +2j-1}, & N>2M, \\ 
  \sum_{j=1}^{(N-M)/2} {\bf w}_{2j-1}-{\bf w}_{N +2j-1}, &N<2M, \end{cases}
$$
$$ 
\quad {\bf v}'' = \sum_{j=1}^{M} {\bf w}_{2j-1} +\sum_{j=1}^{(N-M)/2} {\bf w}_{M+2j-1}, 
$$ 
where in the latter case  the vectors ${\bf w}_j$ for  $j=1,\ldots, N+M-2$ are the shifts of the 
vector ${\bf w}_1$ with 1 in entries 1 and and $M+3$,  $-1$ in entries 3 and $M+1$, and all 
other entries zero. 
\end{lemma} 

\begin{proposition} \label{brankeven} 
For $\gcd (N,M)=1$ with $N+M$ even, the exchange matrices for the T-systems 
(\ref{t1}) and (\ref{t2}) both have rank $N+M-2$. 
\end{proposition} 
\begin{proof}  
To show that $N+M-2$ rows of $B$ are linearly 
independent in each case, we start with (\ref{t1}), where the exchange matrix has size $2N+M$. 
Choosing the first $M$ rows of $B$, and minus the last $M$ rows, together 
with the $N-M-2$ rows ${\bf r}_{N+M+1},\ldots,{\bf r}_{2N-2}$, then by expanding them in the basis ${\bf w}_j$ 
we obtain a square matrix of size $N+M-2$, namely 
\begin{equation}
\label{deltaeven}  
\begin{array}{c@{\!\!\!}l}
\left(
\begin{array}{cccccccccc} 
1 & 0& 1 & \cdots & \cdots & 1 & 0 & 0 & \cdots & \cdots \\ 
0 & 1 & 0 & 1& \cdots  & \cdots & 1 & 0 & \cdots & \cdots \\ 
 & \ddots & \ddots  & \ddots & \ddots &    &  & \ddots &  &  \\ 
 & & \ddots & \ddots & \ddots & \ddots & & & \ddots & \\ 
  &   &         & 0 &  1 & 0 & 1& \cdots & \cdots  & 1 \\ 
1 & 0 & 1 & \cdots & 1 &  0 & \cdots & \cdots & \cdots & \cdots \\ 
 & \ddots & \ddots & \ddots  & & \ddots & \ddots &  & & \\
  & & \ddots & \ddots & \ddots  & & \ddots & \ddots &  & \\
 & & & 1 & 0 & 1 & \cdots & 1 &  0 & 0 
\end{array}
\right)  
& 
\begin{array}[c]{@{}l@{\,}l}
   \left. \begin{array}{c} \vphantom{0}  \\ \vphantom{0}
   \\  \vphantom{\ddots} \\ \vphantom{\ddots} \\ \vphantom{0} \end{array} \right\} & \text{$2M$} \\
\left. \begin{array}{c} \vphantom{0} \\ \vphantom{\ddots}
   \\ \vphantom{\ddots} \\ \vphantom{0}  \end{array} \right\} & \text{$N-M-2$}
\end{array}
\end{array}
\end{equation}
with an alternating block $1 01\cdots 01$ of width $N-M-1$ in each of the first 
$2M$ rows, and a similar block of width $2M-1$ in the last $N-M-2$ rows. Upon expanding about the 
last $2$ columns, we obtain the determinant of the 
Sylvester matrix of size $N+M-4$ for 
the resultant 
$$ 
\mathrm{Res}\left(\frac{x^{2M}-1}{x^2-1}, \frac{x^{N-M}-1}{x^2-1} \right), 
$$ 
which is non-zero if and only if the odd integers 
$N$ and $M$ are coprime. 
For (\ref{t2}), where the exchange matrix has size $2M+N$, we 
again choose the first and the last $M$ rows, and by taking suitable linear  
combinations of the other rows of $B$ we obtain $N-M-2$ more rows, 
leading  to the same matrix (\ref{deltaeven}) as for (\ref{t1}).
\end{proof}

Having verified that the corresponding $B$ matrices have a kernel of the appropriate dimension, Theorem \ref{dimthm} 
now follows directly from the general results on cluster maps in \cite{FH}. However, there remains the question of determining the coefficients $a_k$ that appear as coefficients in the log-canonical bracket (\ref{ubrackets}) for the associated U-system. Although we have not succeeded in finding a simple self-contained expression for the $a_k$, analogous to the formulae 
in Theorem \ref{12bras} for  $c_k,d_k$ that appear in the brackets for the KdV reductions, we will describe a 
simple method which quickly yields  the coefficients of the U-system brackets, and show that 
the log-canonical Poisson structure is unique up to multiplication by a scalar. 
The reduction from the presymplectic form then guarantees that this Poisson structure is nondegenerate. 

We start by considering the U-systems (\ref{1o}) and (\ref{2o}) in Proposition \ref{redsymp}, for fixed coprime integers $N>M$ 
with $N+M$ odd. If we set $m=0$ and take the Poisson bracket of both sides of  
(\ref{1o}) with $u_j$ for $j\in[1,(N+M-3)/2]$, then each value of $j$ gives two 
homogeneous linear equations for the coefficients $a_k$ in (\ref{ubrackets}). Since $a_k=-a_{-k}$, we need only write 
the equations in terms of $a_k$ for positive $k\in[1,N+M-2]$, and taking the brackets with the left-hand side 
of (\ref{1o}) produces the $(N+M-3)/2$ equations 
\beq\label{hgl1} 
a_j + a_{j+1} +\cdots+ a_{N+M-j} =0, \qquad j=2,\ldots, (N+M-1)/2,   
\eeq 
while from the brackets right-hand side the equations split up into the $M-1$ equations 
\beq\label{hgl2}
a_j + a_{j+1}+\cdots + a_{N-M+j-1} =0, \qquad j=1,\ldots, M-1,  
\eeq
together   with a further $(N-M-1)/2$ equations given by 
\beq\label{hgl3}
a_j + a_{j+1}+\cdots + a_{N-M-j} =0, \qquad j=1,\ldots, (N-M-1)/2. 
\eeq
Thus we have a total of $N+M-3$ homogeneous linear equations for $N+M-2$ unknowns, which 
completely determine the  coefficients $a_k$ up to overall multiplication by a scalar. Furthermore, it is 
clear that taking  Poisson brackets of $u_j$ with both sides of  (\ref{2o}) yields an identical set of 
equations, so the two different U-systems (\ref{1o}) and (\ref{2o}) preserve the same log-canonical 
Poisson structure.  

\begin{table*}
\centering
\caption{Index tableau for $N=16$, $M=9$} 
\begin{tabular}{l|c|c|| c| c |r } \hline
& $r$ & $k$ &  $r$ & $k$  &\\  \hline  \hline 
& 1 & 18 &  24 & 7 & $\,\,{\bf II}$\\ 
${\bf II}^{-1}$& 2 & 11 &  23 & 14 & \\ 
${\bf III}$ & 3 & 4  & 22 & 21 & \\ 
& 4 & 22 & 21 & 3 & $\,\,{\bf II}$ \\ 
 & 5 & 15 & 20 & 10 & \\ \hline 
${\bf II}^{-1}$& 6 & 8 & 19 & 17  &\\ 
${\bf III}$&7 & 1 & 18 & (24) &  \\ 
 & 8 & 19 & 17 & 6 & $\,\,{\bf II}$ \\ 
${\bf II}^{-1}$& 9 & 12 & 16 & 13 & \\ 
${\bf III}$& 10 & 5 & 15 & 20 \\ 
& 11 & 23 & 14 & 2 & $\,\,{\bf II}$ \\ 
& 12 & 16 & 13 & 9 & \\ 
\hline
\hline\end{tabular}
\end{table*}

Upon taking linear combinations of the equations, an equivalent set of relations is obtained, namely 
\beq\label{hgliiI} 
{\bf I: } \quad a_k + a_{N+M-k} =0, \qquad k=2,\ldots, (N+M-1)/2. 
\eeq 
(By symmetry, the above is valid for all $k\in[2,N+M-2]$, but here we are concerned with counting the 
number of independent relations.)  
Similarly, subtracting the equations (\ref{hgl2}) from each other in pairs produces the relations 
\beq\label{hgliiiII} 
{\bf II: } \quad a_k = a_{N-M+k}, \qquad k=1,\ldots, M-2, 
\eeq 
while 
 combining the relations (\ref{hgl3}) with the case $j=1$ of (\ref{hgl1}) gives the 
simpler set of equations 
\beq\label{hgli} 
a_k + a_{N-M-k} =0, \qquad k=0,\ldots, (N-M-1)/2,  
\eeq 
which in fact holds for all $k\in [0,N-M]$ by symmetry. Thus we have a simpler set of  
$N+M-3$ independent linear relations, but before proceeding further it is convenient to 
use (\ref{hgli}) followed by  (\ref{hgliiI}) to replace $a_k=-a_{N-M-k}=a_{2M+k}$, so that 
instead of the $(N-M+1)/2$ equations (\ref{hgli}) we can take 
\beq\label{hglipIII} 
{\bf III: } \quad a_k = a_{2M+k}, \qquad k=0,\ldots, (N-M-1)/2.
\eeq 
(The validity of this identity extends to $k\in[0,N-M-2]$, but once again we are counting independent 
relations.)  

Note that from (\ref{hglipIII}) for $k=0$ we have $a_{2M}=0$, and so by (\ref{hgliiI}) with $k=2M$ 
it follows that $a_{N-M}=0$, and similarly $a_k=0$ for any index $k$ that can be related to index $2M$ 
by one of the equations (\ref{hgliiI}), (\ref{hgliiiII}) or (\ref{hglipIII}), while there must be other index values 
with non-zero coefficients, since 
 Theorem \ref{dimthm}  guarantees that a nondegenerate Poisson bracket of the form (\ref{ubrackets}) exists. 
As already mentioned, it turns out that this bracket is unique up to an overall scalar, but the number and location of 
the vanishing coefficients depends on $N$ and $M$ in a complicated way, as does the choice of $\pm$ signs 
on the non-vanishing $a_k$ (if we fix one of them to be 1, say). 

The most efficient way that we have found to calculate the $a_k$ is to write down a tableau of indices 
consisting of four columns, starting with the integers $r=1, \ldots , (N+M-1)/2$ in ascending order (smallest at the top), 
with 
the next column being  index values $k$ ordered according to the rule 
\beq\label{krule} 
k=2Mr \bmod (N+M), 
\eeq  
while the third column consists of the values $r=(N+M-1)/2,\ldots, N+M-1$ 
 in descending order (largest at the top), with the adjacent  
fourth column being  the corresponding values of $k$ given by the rule (\ref{krule}).    
We can then use one of the relations ${\bf I}-{\bf III}$ above to connect index values $k$ in one row 
of the tableau to indices in the same row and/or the next one.

\begin{table*}
\centering
\caption{Index tableau for $N=17$, $M=9$ (odd $k$ only)} 
\begin{tabular}{l|c||  c |l } \hline
&  $k$ &   $k$  & \\  \hline  \hline 
${\bf II}^{-1}$ & 13 & 13 &   \\ 
 & 5 & 21 &  ${\bf III}^{-1}$  \\ 
 & 23 & 3  & ${\bf II}$ \\ 
& 15 & 11 &   \\  \hline  
 & 7 & 19 & ${\bf III}^{-1}$ \\
& (25) & 1 &  ${\bf II}$  \\ 
&17 & 9 &   \\ 
\hline
\hline\end{tabular}
\end{table*}  
 
For illustration, in Table 1 we present the index tableau for the case $N=16$, $M=9$. The index value $N+M-1=24$ 
is included in brackets: this corresponds to a gap in the tableau, since there is no coefficient with this index 
(only the values $k=1,\ldots, 23$ are relevant). 
Observe that index values of $k$ in the same row have coefficients $a_k$ that 
are related to one another by a change of sign, due 
to relation  ${\bf I}$, given by (\ref{hgliiI}). We have added extra columns on the left and right, to indicate 
where an index $k$ in the nearest column is related to the value immediately below it by transformation ${\bf II}$ (sending 
$k\to k+7$ in this case), its inverse   ${\bf II}^{-1}$ (sending 
$k\to k-7$), or by transformation ${\bf III}$ (sending 
$k\to k+18$ here). Note further that  every row contains one of these transformations either on the left or on the right, 
apart from row 5 and the final row; thus a line is inserted under row 5 to signify that it is not related to the row beneath it. 
The tableau implicitly contains 11 ``horizontal'' relations of type ${\bf I}$, relating the values of $k$ on left and right in the same row, this being the number of rows minus one: there is no horizontal relation  in row 7 due to the gap $(24)$. More 
apparent are the ``vertical relations'' in the tableau: there are 7 relations of type   ${\bf II}$ or ${\bf II}^{-1}$, and 3 relations of type ${\bf III}$. Then $11+7+3=21$ leaves one missing relation, namely the fact  that, 
from (\ref{hglipIII}) with $k=0$, the coefficient $a_{18}=0$. All of the index values above the horizontal line in the middle 
are related to the index $k=2M=18$, so have vanishing coefficients, while the non-vanishing coefficients correspond 
to the index values below this line. Hence, if we fix a choice of scale by setting $a_1=1$, then all the plus signs 
are in the lower left part of the tableau, and the minus signs are on the right, so the coefficients in this case 
are obtained as 
$$ 
(a_1,a_2,\ldots, a_{23}) = (1, -1, 0,0,1,-1,0,1,-1,0,0,1,-1,0,0,1,-1,0,1,-1,0,0,1).
$$ 

\begin{theorem}\label{ubrthm} For coprime positive integers $N>M$ with 
$N+M$ odd,  the U-systems (\ref{1o}) and (\ref{2o}) preserve the same 
 log-canonical Poisson bracket (\ref{ubrackets}), which is unique up to overall multiplication by a scalar.  
\end{theorem} 
\begin{proof} The proof follows the same pattern as the example in Table 1. 
The  index tableau contains $(N+M-3)/2$ horizontal relations of type ${\bf I}$, together with  
$M-2$ vertical relations of type ${\bf II}/{\bf II}^{-1}$ and $(N-M-1)/2$ vertical relations of type  ${\bf III}$; 
this makes a total of $(N+M-5)/2$ vertical relations. 
Hence, of the $(N+M-1)/2$ rows in the tableau, in addition to the last row, there is one row somewhere in the middle that is missing a vertical relation, so is not related to the rows beneath it. Hence the indices $k$ in that row and all the 
rows above it, including $k=2M$,  are related to one another and   have coefficients $a_{k}=0$, while the  
relations between the indices in the remaining rows  underneath uniquely determine the non-zero coefficients, up to an overall constant. 
\end{proof} 

\begin{table*}
\centering
\caption{Index tableau for $N=17$, $M=11$ (odd $k$ only)} 
\begin{tabular}{l|c||  c |l } \hline
&  $k$ &   $k$  & \\  \hline  \hline 
 & 25 & 3 & ${\bf II}$  \\ 
 & 19 & 9 &    \\ \hline 
${\bf II}^{-1}$ & 13 & 15  &  \\ 
${\bf II}^{-1}$ & 7 & 21 &   \\    
${\bf III}$ & 1 & (27) &  \\
& 23 & 5 &  ${\bf II}$  \\ 
&17 & 11 &   \\ 
\hline
\hline\end{tabular}
\end{table*}  

The situation for $N+M$ even is slightly more complicated, and the tableau method requires a few modifications 
in that case. 
The  analysis of the analogues of the 
homogeneous linear relations (\ref{hgl1}), (\ref{hgl2}) and (\ref{hgl3}) for the $a_k$ is  slightly trickier. The 
conditions for even indices $k$ decouple from odd $k$, yet the conditions can still be simplified to obtain relations 
of types ${\bf I}$,   ${\bf II}$ and ${\bf III}$, given by the same equations (\ref{hgliiI}),  (\ref{hgliiiII}) and   (\ref{hglipIII}), 
respectively, but with slightly different ranges of indices, namely 
\beq\label{eventypes} 
{\bf I}: \,\, k=3,\ldots,(N+M)/2,\,\, {\bf II}: \,\, k=0,\ldots,M-3,  \,\,{\bf III}: \,\,k=1,\ldots,(N-M)/2 -1,
\eeq 
together with one additional relation: 
\beq\label{extra} 
 a_{(N-M)/2}=0. 
\eeq 
This gives a total of $N+M-4$ homogeneous linear equations for the coefficients $a_k$ with indices $k\in[1,N+M-3]$.  
There are slight differences in the analysis according to whether $(N+M)/2$ is odd or even, so 
for illustration we describe one example of each before discussing the general case.

In Table 2 we present the example $N=17$, $M=9$,  where (for reasons that will become clear) we have given the tableau for odd indices $k$ only, with the column labels $r$ omitted. In this example,  
$(N+M)/2=13$ is odd. Due to the relation ${\bf I}$ with $k=(N+M)/2=13$, 
it follows that $a_{13}=0$, so we start the first row with this value of $k$ appearing twice. In the first column we obtain 
each entry from the one above by subtracting $N-M=8$ and evaluating modulo $N+M$, that is, $\bmod 26$, 
while in descending the second column we add $8$ instead of subtracting and evaluate $\bmod 26$ until all 12 
of the odd indices $1,3,\ldots, 23$ have been entered; and there is a gap in the tableau corresponding to one spurious index $(25)$. Then as before we write ${\bf II}$, ${\bf III}$ or there inverses next to a column if this transformation relates an  entry to the one below it. In this case we find that the indices fourth row are unrelated to the entries below, so there 
is a line under this row. There are 6 horizontal relations in the tableau (type ${\bf I}$) and 5 vertical relations 
(written as type ${\bf II}/{\bf III}$ or their inverses), giving a total of 11 relations for the odd index 
coefficients.  Because  $a_{13}=0$, it follows that all the indices in the first four rows correspond to vanishing 
coefficients, while the entries in the last three rows below the line correspond to the non-vanishing coefficients. Upon fixing $a_1=1$, this determines the coefficients with odd indices as 
$$
(a_1,a_3,\ldots, a_{23})= (1,0,0-1,1,0,0,0,-1,1,0,0). 
$$
For the even indices, note that there are 11 coefficients $a_2,a_4,\ldots,a_{22}$, while the total number of relations 
remaining is also $22-11=11$. Hence the homogeneous linear system for the even index coefficients has only the zero 
solution, 
$$ 
a_{2j}=0, \qquad j=1,\ldots , 11, 
$$ 
and there is no need to present the even indices in a tableau. 
 
In contrast, Table 3 is the odd index tableau for the example  $N=17$, $M=11$, in which case 
$(N+M)/2=14$ is even. This means that the index $(N-M)/2=3$ is now odd, so we place this index in the top 
right entry and descend in steps of $N-M=6$, adding in the right column, subtracting in the left column, and 
evaluating modulo $N+M$, i.e.\ $\bmod 28$. For the 13 odd indices, there are 6 horizontal relations, and 5 vertical relations, 
plus the extra relation (\ref{extra}), with no vertical relation in the second row. Hence the first two rows 
correspond to vanishing coefficients, and the remaining rows contain the indices of non-vanishing ones. Fixing 
$a_1=1$ as before, this determines the odd index coefficients as 
$$   
(a_1,a_3,\ldots, a_{25}) = (1,0,-1,1,0,-1,1,-1,1,0,-1,1,0). 
$$
There are  $24-12=12$ remaining homogeneous 
relations for 12 even indices, so again the even index coefficients are all zero.  

\begin{theorem}\label{ubrthm2} For coprime positive integers $N>M$ with 
$N+M$ even,  the U-systems (\ref{1e}) and (\ref{2e}) preserve the same 
 log-canonical Poisson bracket (\ref{ubrackets}), which is unique up to overall multiplication by a scalar.  
\end{theorem} 
\begin{proof} Using (\ref{eventypes}) and  (\ref{extra}), 
the  precise counting of the odd/even index relations is slightly different according to whether  $(N+M)/2$ is odd or even, 
but in both cases there are $(N+M)/2-2$ independent relations for odd indices, and the same number for even indices. Since 
there are  $(N+M)/2-2$ even index coefficients, the corresponding homogeneous linear system has only the zero solution. 
The tableau  for the odd indices has one row in the middle with no vertical relation,  with the indices of 
vanishing coefficients 
in this row and above, and indices for non-vanishing coefficients below, 
giving the unique solution for the $(N+M)/2-1$ odd index coefficients (up to an overall constant).  
\end{proof} 

\section{Proof of Lemma \ref{lemmacoefmap1}} \label{appendix1}

We set $v_{N+M}=v_0+\alpha(\frac{1}{v_{N}}-\frac{1}{v_{M}}))$, where $v_m$, $0 \leq m \leq N+M-1$, are given by 
\eqref{vfromu1}. Assume that 
$$\{v_0,v_{N+M-m} \}_u = f(v_0,v_1 \dots, v_{N+M-1}),$$
where $f$ is defined by \eqref{recP1}. Since the map $\phi_u^1$ is Poisson, 
$$\{v_m,v_{N+M} \}_u = f(v_m,v_{m+1}, \dots, v_{m+N+M-1})=:F(v_0,v_1 \dots, v_{N+M-1}),$$
where $v_{N+M},v_{N+M+1}, \dots,  v_{m+N+M-1}$ have been written with respect to $v_0, \dots, v_{N+M-1}$ 
using the recurrence \eqref{RedkdV}. 
Equivalently,   
\begin{equation} \label{lemeqbr}
\{v_m,v_0 \}_u -\frac{\alpha}{v_{N}^2} \{ v_m,v_N \}_u+\frac{\alpha}{v_{M}^2}\{v_m,v_M \}_u=F(v_0,v_1, \dots, v_{N+M-1}).
\end{equation} 
We will use this equation to prove \eqref{cmeq}. 

For $\alpha=0$,  $f(v_0,v_1, \dots, v_{N+M-1})=c_{N+M-m} v_0 v_{N+M-m}$. So, 
$$ F(v_0, \dots, v_{N+M-1}):= f(v_m, \dots, v_{m+N+M-1})=c_{N+M-m} v_m v_{N+M}=c_{N+M-m} v_0 v_m.$$ 
Also, $\{v_0,v_m \}_u=c_m v_0 v_m$, and from \eqref{lemeqbr} 
we derive 
\begin{equation} \label{cnm1}
c_m=-c_{N+M-m}, \ 0<m<N+M. 
\end{equation}

For $0<m<M$, we have $N<N+M-m<N+M$, $0<M-m<M$ and $N-M<N-m<N$. From \eqref{recP1} we obtain  
 $$f(v_0,v_1, \dots, v_{N+M-1})=c_{N+M-m} v_0 v_{N+M-m}+c_{N+M-m} \alpha \frac{v_0}{v_{M-m}}+c_{M-m} \alpha \frac{v_0}{v_{M-m}}.$$
 So, 
 \begin{eqnarray*}
F(v_0, \dots, v_{N+M-1}) &:=& f(v_m, \dots, v_{m+N+M-1}) \\ &=&c_{N+M-m} v_m v_{N+M}+c_{N+M-m} \alpha \frac{v_m}{v_{M}}+c_{M-m} \alpha \frac{v_m}{v_{M}} \\ 
 &=& c_{N+M-m} v_0 v_m +c_{N+M-m} \alpha \frac{v_m}{v_N}+c_{M-m} \alpha \frac{v_m}{v_M} 
\end{eqnarray*}
and $\{v_0,v_m \}_u=c_m v_0 v_m$, $\{v_m,v_N \}_u=c_{N-m} v_m v_N$, $\{v_m,v_M \}_u=c_{M-m} v_m v_M$. Hence, from \eqref{lemeqbr}  
we get 
\begin{equation} \label{cnm2}
c_m=-c_{N+M-m}=c_{N-m}, \ 0<m<M. 
\end{equation}

For $M<m<N$, we have $M<N+M-m<N$, $0<m-M<N-M$ and $0<N-m<N-m$. Now, from \eqref{recP1} we obtain  
 $$f(v_0,v_1, \dots, v_{N+M-1})=c_{N+M-m} v_0 v_{N+M-m}.$$ So, 
 \begin{eqnarray*}
F(v_0, \dots, v_{N+M-1})=c_{N+M-m} v_m v_{N+M}  
 = c_{N+M-m} v_0 v_m +c_{N+M-m} \alpha \frac{v_m}{v_N}-c_{N+M-m} \alpha \frac{v_m}{v_M}. 
\end{eqnarray*}
Also, $\{v_0,v_m \}_u=c_m v_0 v_m$, $\{v_m,v_N \}_u=c_{N-m} v_m v_N$, $\{v_m,v_M \}_u=-c_{m-M} v_m v_M$ and from \eqref{lemeqbr}  
we derive 
\begin{equation} \label{cnm3}
c_m=-c_{N+M-m}=-c_{m-M}=c_{N-m}, \ M<m<N. 
\end{equation}

Finally, for $N<m<N+M$, we have $0<N+M-m<M$, $N-M<m-M<N$ and $0<m-N<M$. Here as well,  
 $f(v_0,v_1, \dots, v_{N+M-1})=c_{N+M-m} v_0 v_{N+M-m}$ and  
 \begin{eqnarray*}
F(v_0, \dots, v_{N+M-1})=c_{N+M-m} v_0 v_m +c_{N+M-m} \alpha \frac{v_m}{v_N}-c_{N+M-m} \alpha \frac{v_m}{v_M}. 
\end{eqnarray*}
Now, 
\begin{eqnarray*}
\{v_0,v_m \}_u &=& c_m v_0 v_m+ c_m \alpha \frac{v_0}{v_{m-N}}+\frac{\alpha}{v_{m-N}^2} \{ v_0,v_{m-N} \}_u \\
&=&
c_m v_0 v_m+ c_m \alpha \frac{v_0}{v_{m-N}}+c_{m-N} \alpha \frac{v_0}{v_{m-N}},
\end{eqnarray*}
$\{v_m,v_N \}_u=-c_{m-N} v_m v_N$ and $\{v_m,v_M \}_u=-c_{m-M} v_m v_M$.  So, equation \eqref{lemeqbr}  
implies 
\begin{equation} \label{cnm4}
c_m=-c_{N+M-m}=-c_{m-M}=-c_{m-N}, \ N<m<N+M. 
\end{equation}

Now, taking into account that $c_{-m}=-c_m$, for $0<m \leq N+M-1$, we can combine (\ref{cnm1}--\ref{cnm4}) 
to find 
\begin{equation*}
c_m=-c_{N+M-m}=-c_{m-N}=-c_{m-M}, \ \ 0<m \leq N+M-1, 
\end{equation*}
while from  \eqref{cnm1} and \eqref{cnm4}, for $0<m<M$, we obtain $c_{M-m}=-c_{N+m}=c_{N+m-N}=c_m$.
\qed

\section{Proof of Lemma \ref{lemmacoefmap2}} \label{appendix2}
Following the proof of lemma \ref{lemmacoefmap1} 
we consider 
$$\{v_0,v_{N+M-m} \}_u = g(v_0,v_1 \dots, v_{N+M-1}),$$
where here $g$ is defined by \eqref{recP2} (the function g here must not be confused with the g-variables introduced in section \ref{sectiongvar}). Since the map $\phi_u^2$ is Poisson, 
$$\{v_m,v_{N+M} \}_u = g(v_m,v_{m+1}, \dots, v_{m+N+M-1})=:G(v_0,v_1 \dots, v_{N+M-1}),$$
where $v_{N+M},v_{N+M+1}, \dots,  v_{m+N+M-1}$ are evaluated with respect to $v_0, \dots, v_{N+M-1}$ 
using the recurrence \eqref{RedkdV}. 
So,   
\begin{equation} \label{lemeqbr2}
\{v_m,v_0 \}_u -\frac{\alpha}{v_{N}^2} \{ v_m,v_N \}_u+\frac{\alpha}{v_{M}^2}\{v_m,v_M \}_u=G(v_0,v_1, \dots, v_{N+M-1}).
\end{equation} 

First, by considering $\alpha=0$ we can show immediately, as in the case of Lemma \ref{lemmacoefmap1},  
that 
\begin{equation} \label{dnm2}
d_m=-d_{N+M-m}, \ 0<m<N+M. 
\end{equation}

Now, for $N-M<m \leq N$, we have $M \leq N+M-m <2M$ and $0 \leq N-m <M$. 
In this case $$g(v_0,v_1 \dots, v_{N+M-1})= d_{N+M-m} v_0 v_{N+M-m}-d_{N+M-m} \alpha \frac{v_0}{v_{N-m}}-d_{N-m} \alpha \frac{v_0}{v_{N-m}}.$$ 
So, 
\begin{eqnarray*}
G(v_0,v_1, \dots, v_{N+M-1}) &=& g(v_m,v_{m+1}, \dots, v_{m+N+M-1}) \\ 
&=& d_{N+M-m} v_m v_{N+M}-d_{N+M-m} \alpha \frac{v_m}{v_{N}}-d_{N-m} \alpha \frac{v_m}{v_{N}} \\
&=& d_{N+M-m} v_0 v_m-d_{N+M-m} \alpha \frac{v_m}{v_M}-d_{N-m} \alpha \frac{v_m}{v_N}
\end{eqnarray*}
and $\{v_m,v_N \}_u=d_{N-m} v_m v_N$.   
Hence, from \eqref{lemeqbr2} and \eqref{cnm2} we derive
\begin{equation} \label{case2A}
-\{v_0,v_m \}_u+ \frac{\alpha}{v_M^2}\{v_m,v_M \}_u=-d_m v_0 v_m+d_m \alpha \frac{v_m}{v_M}.
\end{equation}
If $0<m<M$, equation \eqref{case2A} is equivalent to 
$$-d_m v_0 v_m+ d_{M-m}  \alpha \frac{v_m}{v_M}=-d_m v_0 v_m+d_m  \alpha \frac{v_m}{v_M}.$$
So, $d_m=d_{M-m}$. \\ 
If $M \leq m \leq N$, equation \eqref{case2A} is equivalent to 
\begin{equation} \label{case2B}
  -\frac{\alpha}{v_M^2}\{v_M,v_m \}_u+  \frac{\alpha}{v_{m-M}^2}\{v_0,v_{m-M} \}_u= d_m  \alpha \frac{v_m}{v_M}-d_m  \alpha \frac{v_0}{v_{m-M}}.
\end{equation}
But, $$\{v_0,v_{m-M} \}_u= d_{m-M} v_0 v_{m-M}+K(v_0,v_{m-2M}, v_{m-3M}, \dots v_{m-kM}),$$ for some $k \geq 2$,  
where the function $K$ is determined by the recurrence \eqref{recP2} . Subsequently, 
$\{v_M,v_m \}_u=d_{m-M} v_M v_{m}+K(v_M,v_{m-M}, v_{m-2M}, \dots v_{m-(k-1)M})$. Substituting these to 
\eqref{case2B}, we derive $d_m=-d_{m-M}$. 
Finally, taking into account that $d_{-m} = -d_m$ we have 
$$ d_m=-d_{m-M},  \ \text{for} \  N-M<m \leq N.$$
Equivalently, if we set  $m=N+M-l$, by the last equation and \eqref{dnm2} we get 
\begin{equation} \label{secinterv}
d_l=-d_{N+M-l}=d_{N-l}=-d_{l+M}, \ M \leq l <2M. 
\end{equation}

Furthermore, for $N<m<N+M$, we have $0<N+M-m<M$  and $0<m-N<M$. 
Here, 
$g(v_0,v_1 \dots, v_{N+M-1})= d_{N+M-m} v_0 v_{N+M-m}$,
\begin{eqnarray*}
G(v_0,v_1, \dots, v_{N+M-1}) = d_{N+M-m} v_0 v_m-d_{N+M-m} \alpha \frac{v_m}{v_M}+d_{N+M-m} \alpha \frac{v_m}{v_N} 
\end{eqnarray*}
and in a similar way, from \eqref{lemeqbr2} we derive
\begin{equation*} \label{lastint}
d_m=-d_{m-N}=-d_{m-M}, \  N<m<N+M 
\end{equation*} 
or by setting $m=N+M-l$, 
\begin{equation} \label{secinterv2}
d_l=-d_{N+M-l}=d_{M-l}=d_{N-l}=-d_{l+M}, \ \text{for} \ 0 < l <M. 
\end{equation}

To sum up, from \eqref{secinterv} and \eqref{secinterv2} we have shown that 
\begin{equation} \label{final02M}
d_m=-d_{m+M}=-d_{N+M-l}, \ \text{for} \ 0 < m < 2M. 
\end{equation}

Finally, we notice that for $M \leq m <N+M$
\begin{eqnarray} \label{dmd2M}
d_m= 
- \sum \limits_{i=0}^{N-1} \sum \limits_{j=m-M}^{m-1} a_{j-i}
=- \sum \limits_{i=0}^{N-1} \sum \limits_{j=m+M}^{m+2M-1} a_{j-i}=d_{m+2M},  
\end{eqnarray}
since, by (\ref{hglipIII} ),  the coefficients $a_m$ of the Poisson bracket \eqref{ubrackets} satisfy 
the equation $a_k=a_{k+2M}$ (extended to negative indices via $a_{-k}=-a_k$).  

From \eqref{final02M} and \eqref{dmd2M}, we conclude that 
$$d_m=-d_{N+M-m}=-d_{m-M}$$
for $0 < m <N+M$, which completes the proof. 
\qed

\end{document}